\documentclass[aps, superscriptaddress]{revtex4}
\newcommand{\bra}[1]{\langle #1|}
\newcommand{\ket}[1]{|#1\rangle}

\usepackage{amsmath,amsthm,amsfonts,amssymb} 
\usepackage{fontenc}
\usepackage[all]{xy}
\usepackage{graphicx,subfigure,multirow}
\usepackage{tikz}
\usetikzlibrary{shapes, arrows, shadows}
\usepackage[latin1]{inputenc}
\usepackage[scaled]{helvet}
\usepackage[toc, page]{appendix}
\usepackage{graphicx}				
\usepackage{amssymb}
\usepackage{changes}
\usepackage{lipsum}
\newtheorem{definition}{Definition}
\newtheorem{theorem}{Theorem}
\newtheorem{proposition}{Proposition}
\newtheorem{lemma}{Lemma}
\newtheorem{Corollary}{Corollary}
\newtheorem{Example}{Example}

\begin{document}

\title{Self-testing through EPR-steering}   
\author{Ivan \v{S}upi\'{c}}
\email{ivan.supic@icfo.es}
\affiliation{ICFO-Institut de Ciencies Fotoniques, The Barcelona Institute of Science and Technology, 08860 Castelldefels (Barcelona), Spain.}
\author{Matty~J. Hoban}
\email{matthew.hoban@cs.ox.ac.uk}
\affiliation{University of Oxford, Department of Computer Science, Wolfson Building, Parks Road, Oxford OX1 3QD, UK.}
\affiliation{School of Informatics, University of Edinburgh, 10 Crichton Street, Edinburgh EH8 9AB, UK.}					

\begin{abstract}
The verification of quantum devices is an important aspect of quantum information, especially with the emergence of more advanced experimental implementations of quantum computation and secure communication. Within this, the theory of device-independent robust self-testing via Bell tests has reached a level of maturity now that many quantum states and measurements can be verified without direct access to the quantum systems: interaction with the devices is solely classical. However, the requirements for this robust level of verification are daunting and require high levels of experimental accuracy. In this paper we discuss the possibility of self-testing where we only have direct access to one part of the quantum device. This motivates the study of self-testing via EPR-steering, an intermediate form of entanglement verification between full state tomography and Bell tests. Quantum non-locality implies EPR-steering so results in the former can apply in the latter, but we ask what advantages may be gleaned from the latter over the former given that one can do partial state tomography? We show that in the case of self-testing a maximally entangled two-qubit state, or ebit, EPR-steering allows for simpler analysis and better error tolerance than in the case of full device-independence. On the other hand, this improvement is only a constant improvement and (up to constants) is the best one can hope for. Finally, we indicate that the main advantage in self-testing based on EPR-steering could be in the case of self-testing multi-partite quantum states and measurements. For example, it may be easier to establish a tensor product structure for a particular party's Hilbert space even if we do not have access to their part of the global quantum system.
\end{abstract}

\maketitle

\section{Introduction}

The certification of quantum devices is an important strand in current research in quantum information. Research in this direction is not only of relevance to quantum information but also the foundations of quantum theory: what are the truly quantum phenomena? For example, if presented with devices as black boxes that are claimed to contain systems associated with particular quantum states and measurements, we can certify these claims by demonstrating quantum non-locality, i.e. by violating a particular Bell inequality \cite{Bell}. 

The obvious aspect of quantum non-locality that is useful for quantum information is that it can certify quantum entanglement. While this is relevant for the certification of the presence of quantum entanglement, if we wish to certify a particular state and measurement we need more information. More specifically, given a particular violation of a Bell inequality, can we infer the state and measurements? The amount of information necessary to certify a particular state once entanglement is certified has been discussed in Ref. \cite{Carmeli}. Let us consider the specific example of the Clauser-Horne-Shimony-Holt (CHSH) inequality \cite{CHSH}. It can be shown that (up to local operations that will be specified later) the only state that can maximally violate the CHSH inequality is the maximally entangled two-qubit state \cite{CHSHUniqueness}. Furthermore, if we are close to the maximal violation, then we are also close to this maximally entangled state (for appropriate notions of closeness) \cite{MKYS}. Results in this direction are referred to as \textit{ robust self-testing} (RST) such that a near-maximal violation of a Bell inequality robustly self-tests a state. We can also robustly self-test measurements performed on a state therefore equipping us with certification techniques for both states and measurements.

To be more concrete, RST is possible if the correlations we observe in a Bell test are $\epsilon$-close to some ideal correlations -- such as those maximally violating a Bell inequality -- then we can infer that the state used in the Bell test is $O(\sqrt{\epsilon})$-close to our ideal state. The notion of closeness will be expounded upon later but for correlations we often talk about the difference between the maximal Bell inequality violation and the violation obtained in the experiment, and for quantum states, we refer to the trace distance. This quadratic difference in the distance measures cannot be improved upon if we only have access to the correlations \cite{BQC}.

In this direction, a bounty of results have emerged. There are now analytical methods for robustly self-testing Greenberger-Horne-Zeilinger (GHZ) states \cite{PVN}, graph states \cite{McKague}, partially entangled two-qubit states \cite{Cedric} and the so-called W state \cite{Wu}. In addition to this, numerical robust self-testing methods were developed that allow for using arbitrary Bell inequalities \cite{Bancal}. Also, it is worth noting that by simply and directly considering the correlations produced in the experiment, numerical methods developed in Refs. \cite{Bancal,Olmo,YVBSN} can also be tailored to these considerations.

It is now well-established that the violation of a Bell inequality is not the only method for detecting entanglement in general. It is the appropriate method if one only has access to measurement statistics, i.e. the devices are treated like black boxes. Clearly, if we have direct access to the quantum state (e.g. the devices are trusted), we can do full state tomography to see if it is an entangled state. There does exist a third option, if a provider claims to produce a bipartite entangled state and sends one half of the state to a client who wants to use the state. We can assume that the client trusts all of the apparatus in their laboratory and can thus do state tomography on their share of the system. This set-up corresponds to the notion of \textit{EPR-steering} in the study of entanglement \cite{Erwin,WJD}, where EPR represents Einstein-Podolsky-Rosen in tribute to their 1935 original paper \cite{EPR}. A natural question is whether one can perform robust self-testing in such a scenario? This is obviously true since we can use the violation of a Bell inequality between the client and provider. A better question is whether it is vastly more advantageous to consider self-testing in this scenario? In this work, we address this question.

Before describing the work in this paper, we would like to motivate this scenario from the point-of-view of quantum information. In particular, studying these EPR-steering scenarios may be useful when considering \textit{Blind Quantum Computing} where a client has restricted quantum operations and wishes to securely delegate a computation to a ``server" that has a full-power quantum computer \cite{BFK,BQC}. By securely, we mean that the server does not learn the input to the computation nor the particular computation itself. In this framework, the client trusts all of his quantum resources but distrusts the server. EPR-steering has also been utilised for \textit{one-sided device-independent quantum key distribution} where the ``one-sided" indicates that one of the parties does not trust their device but the other does \cite{SteerQKD,Walk}. There have even been experimental demonstrations of cryptographic schemes in this direction \cite{Gehring}. Also in this one-sided device-independent approach, the detection loophole is less detrimental to performing cryptographic tasks as compared with full device-independence so it is more amenable to current optical experiments \cite{Wittmann, Armstrong}.

Since one party (the client) now trusts all systems in their laboratory, they can perform quantum state tomography; after all, they know the Hilbert space dimension of their quantum systems and can choose to make measurements that characterise states of that particular dimension. This novel aspect of EPR-steering as compared to standard non-locality introduces a novel object of study, called the assemblage: the reduced states on a client's share of some larger states conditioned on measurements made on the provider's side \cite{assemblage}. An element of an assemblage is then a sub-normalized quantum state and we can now also phrase robust self-testing in terms of these objects, which we call \textit{robust assemblage-based one-sided self-testing} (AST) with ``one-sided" to indicate there is one untrusted party. In essence, we show that AST can be achieved and the experimental state is at least $O(\sqrt{\epsilon})$-close to an ideal state if the observed elements of an assemblage are $\epsilon$-close to the ideal elements (where distance in both cases is the trace distance). This is in addition to considering the correlations between the client and provider obtained from performing a measurement on the elements of an assemblage, which we call \textit{robust correlation-based one-sided self-testing} (CST) -- the notions of robustness are the same as for RST.

Conventional RST based on Bell inequality violation implies CST so in the latter scenario we will never do any worse than in the former. Furthermore, CST implies AST so the latter truly captures the novel capabilities in the formalism. In this work, for particular situations we show both analytically and numerically that one can do better in the framework of CST and AST as compared to current methods in RST. This is to be expected since by trusting one side, we should have access to more information about our initial state. On the other hand, we show that the degree of the improvement is not as dramatic as we would like. In particular, if the assemblage is, in some sense, $\epsilon$-close to the ideal assemblage, we can only establish $O(\sqrt{\epsilon})$-closeness of our operations to the ideal case. This quadratic difference is also shown to be a general limitation and not just a limitation of our specific methods. In this way, from the point-of-view of self-testing, EPR-steering behaves much like quantum non-locality.

We indicate where AST and CST could also prove advantageous over RST and this is in the case of establishing the structure of sub-systems within multi-partite quantum states. That is, in certain RST proofs a lot of work and resources goes into establishing that untrusted devices have quantum systems that are essentially independent from one another. In addition to considering the self-testing of a bipartite quantum state, we show that one can get further improvements by establishing a tensor product structure between sub-systems. This could be where the essential novelties of AST and CST lie.

Aside from work in the remit of self-testing there is other work in the direction of entanglement verification between many parties. For example, Pappa \textit{et al} show how to verify GHZ states among $n$ parties if some of them can be trusted while others not \cite{Pappa}. Their verification proofs boil down to establishing the probability with which the quantum state passes a particular test given the state's distance from the ideal case. This can be seen as going in the other direction compared to CST, where we ask how close a state is to ideal if we pass a test (demonstrating some ideal correlations) with a particular probability. Our work thus nicely complements some of the existing methods in this direction.

Another line of research that is related to our own is to characterise (non-local) quantum correlations given assumptions made about the dimension of the Hilbert space for one of the parties \cite{Gonzalo}. This assumption of limiting the dimension is a relaxation of the assumption that devices in one of the parties' laboratories are trusted. These works are relevant for \textit{semi-device-independent quantum cryptography} and \textit{device-independent dimension witnesses} \cite{BrunnerPawlowski,Gallego}

In Sec. \ref{sec1} we outline the general framework, introduce CST and AST and introduce the methods which will be relevant. Given our framework, in Sec. \ref{sec2} we demonstrate how to self-test the maximally entangled two-qubit state and give analytical and numerical results demonstrating an improvement over conventional RST. In Sec. \ref{sec3} we briefly discuss the self-testing of multi-partite states and give numerical results showing how the GHZ state can be self-tested. We also discuss how one could exploit tensor product structure on the trusted side to aid self-testing. We conclude with some general discussion in Sec. \ref{sec4}.

\section{General Set-up}\label{sec1}

In this section we introduce the framework in which our results will be cast. For brevity we will restrict ourselves to the case of two parties each with access to some devices. In Sec. \ref{sec3} we will extend the framework to more-than-two parties. In our setting (see Fig. \ref{fig:fig1}), one of the parties is the client and the other is the provider and the two of them share both quantum and classical communication channels and all devices are assumed to be quantum mechanical. Therefore we can associate the parties with the finite-dimensional Hilbert spaces $\mathcal{H}_{C}$ and $\mathcal{H}_{P}$ for the client and provider respectively \footnote{We assume finite dimensional Hilbert spaces for our purposes since we want to self-test systems of finite dimension. We can follow Reichardt, Unger and Vazirani and allow for finite dimensional systems approximating those of infinite dimension since robustness allows for this \cite{BQC}.}. The quantum communication channel is used to send a quantum system from the provider to the client and the client will then perform tomography on this part of the state. After the provider has communicated a quantum system, there will be some joint quantum system and the client can now ask the provider (using the classical communication channel) to perform measurements on their share of the system; the outcome is then communicated to the client. 

\begin{figure}
  \centering
    \includegraphics[width=0.33\textwidth]{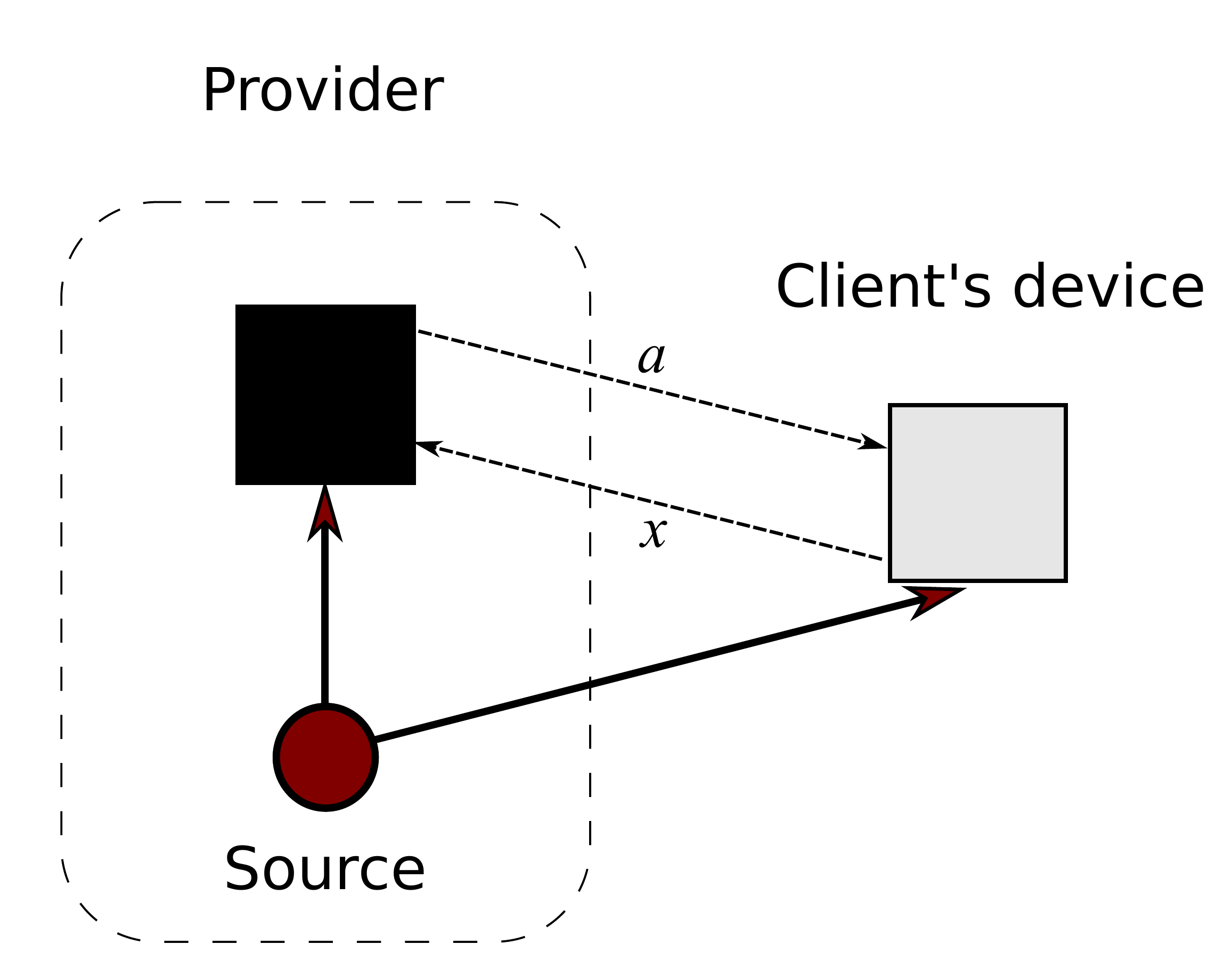}
    \caption{In our framework we have a client who has direct access to his part of the quantum system generated by the source in the provider's laboratory. We can also ask the provider to perform a measurement labelled by $x$ and generate an outcome labelled by $a$ all the while treating the provider's measurement device and the source as a black box. The dotted lines denote classical channels, while full lines represents a quantum channel.\label{fig:fig1}}
\end{figure}

In this work we assume that the provider gives the client arbitrarily many copies of the subsystem such that they can do perfect tomography on their quantum system. We will not consider complications introduced by only having access to finitely many systems. This is a standard assumption in many works on self-testing and we will comment on relaxing this assumption in Sec. \ref{sec4}.

After the provider sends a quantum system to the client they share a quantum state $\rho_{CP}$, a density matrix acting on the Hilbert space $\mathcal{H}_{C}\otimes\mathcal{H}_{P}$. Crucially, in our work, the dimension of the Hilbert space $\mathcal{H}_{C}$ is known but the space $\mathcal{H}_{P}$ can have an unrestricted dimension since we do not, in general, trust the provider. Therefore, without loss of generality, the density matrix $\rho_{CP}=\ket{\psi}\bra{\psi}$ is associated with a pure state $\ket{\psi}\in\mathcal{H}_{C}\otimes\mathcal{H}_{P}$ since we can always dilate the space $\mathcal{H}_{P}$ to find an appropriate purification.

After establishing the shared state $\ket{\psi}$, the client asks the provider to perform a measurement from a choice of possible measurements. These measurements are labelled by a symbol $x\in\{0,1,2,...,(d-1)\}$ if there are $d\in\mathbb{N}$ possible choices of measurement. For each measurement, there are $k\in\mathbb{N}$ possible outcomes labelled by the symbol $a\in\{0,1,2,...,(k-1)\}$. The client then communicates a value of $x$ to the provider and then receives a value of $a$ from the provider. Again, since the dimension of $\mathcal{H}_{P}$ is unrestricted, we assume that the measurement made by the provider has outcomes that are associated with projectors $E_{a|x}$ such that $\sum_{a}E_{a|x}=\mathbb{I}$ and $E_{a|x}E_{a'|x}=\delta_{a,a'}E_{a|x}$. 

Conditioned on each measurement outcome $a$ given the choice $x$, the client performs state tomography on their part of the state $\ket{\psi}$ which can be described in terms of the operators $\sigma_{a|x}=\textrm{tr}_{P}\left(\mathbb{I}_{C}\otimes E_{a|x}\ket{\psi}\bra{\psi}\right)$ where $\mathbb{I}_{C}$ is the identity operator acting on $\mathcal{H}_{C}$ and $\textrm{tr}_{P}\left(\cdot\right)$ is the partial trace over the provider's system. An \textit{assemblage} is then the set $\{\sigma_{a|x}\}_{a,x}$ with elements satisfying $\sum_{a}\sigma_{a|x}=\textrm{tr}_{P}(\ket{\psi}\bra{\psi})=\rho_{C}$, the reduced state of the client's system. One can extract the probability $p(a|x)$ of the provider's measurement outcome $a$ for the choice $x$ by taking $\textrm{tr}(\sigma_{a|x})=p(a|x)$.
 
Instead of studying the assemblage directly, we may simplify matters by considering the \textit{correlations} between the client and provider where both parties make measurements and look at the conditional probabilities $p(a,b|x,y)$ where $y\in\{0,1,...,(d-1)\}$ is the client's choice of measurement and $b\in\{0,1,2,...,(k-1)\}$ the outcome for that choice. If the measurement made by the client is described in terms of the generalised measurement elements $F_{b|y}$ such that $\sum_{b}F_{b|y}=\mathbb{I}_{C}$ then these correlations can be readily obtained from elements of the assemblage as $p(a,b|x,y)=\textrm{tr}\left(F_{b|y}\sigma_{a|x}\right)$.

In self-testing, the provider claims that they are manufacturing a particular state $\ket{\tilde{\psi}}\in\mathcal{H}_{C}\otimes \mathcal{H}'_{P}$ and performing particular (projective) measurements $\{\tilde{E}_{a|x}\}_{a,x}$ on $\mathcal{H}'_{P}$. We call this combination of state and measurements the \textit{reference experiment} to distinguish it from the physical experiment where $\ket{\psi}$ and $\{E_{a|x}\}_{a,x}$ are the state and measurements respectively. Since we do not have direct access to the Hilbert space of the provider it is possible that they are manufacturing something different that has no observable effect on experimental outcomes. For example, they could prepare the state $\ket{\psi}=\ket{\tilde{\psi}}\ket{0}$ and retain the system in state $\ket{0}$ but never perform any operation on it. This will not affect the assemblage so we must allow for operations on the provider's system in $\mathcal{H}_{P}$ that leave assemblages unaffected. Following the discussion by McKague and Mosca, some of these changes include \cite{MkM}:
\begin{enumerate}
\item Unitary change of basis in $\mathcal{H}_{P}$
\item Adding ancillae $\ket{\mathcal{A}}$ to physical systems (in tensor product) upon which measurements do not act, i.e. $\ket{\psi}\rightarrow\ket{\psi}\ket{\mathcal{A}}$
\item Altering the measurements $\{E_{a|x}\}_{a,x}$ outside the support of the state $\ket{\psi}$
\item Embedding the state $\ket{\psi}\in\mathcal{H}_{C}\otimes\mathcal{H}_{P}$ and measurements $\{E_{a|x}\}_{a,x}$ into a Hilbert space $\mathcal{H}_{C}\otimes\mathcal{K}_{P}$ where $\mathcal{K}_{P}$ has a different dimension to $\mathcal{H}_{P}$. 
\end{enumerate}
Allowing for these possible transformations we need an appropriate notion of equivalence between the physical experiment and the reference experiment. We say that the physical experiment associated with the state $\ket{\psi}$ and measurements $\{E_{a|x}\}_{a,x}$ are equivalent to the reference experiment associated with the state $\ket{\tilde{\psi}}$ and measurements $\{\tilde{E}_{a|x}\}_{a,x}$ if there exists an isometry $\Phi:\mathcal{H}_{P}\rightarrow\mathcal{H}_{P}\otimes\mathcal{H}'_{P}$ such that 
\begin{align}\label{equiv}
\Phi(\ket{\psi})&=\ket{\mathcal{A}}\ket{\tilde{\psi}},\nonumber\\
\Phi(\mathbb{I}_{C}\otimes E_{a|x}\ket{\psi})&=\ket{\mathcal{A}}\left(\mathbb{I}_{C}\otimes \tilde{E}_{a|x}\right)\ket{\tilde{\psi}},
\end{align}
for all $a$, $x$ and $\ket{\mathcal{A}}\in\mathcal{H}_{P}$. 

A consequence of this notion of equivalence is that if a physical experiment is equivalent to the reference experiment then the former can be constructed from the latter by the operations described above. In the other direction, if the provider does indeed construct the reference experiment and then performs one of the transformations listed above then an isometry can always be constructed to establish equivalence between the physical and reference experiments. An important issue in self-testing based on probabilities is that experimental probabilities are invariant upon taking the complex conjugate of both the state and measurements. Thus, the best one can hope for in this kind of self-testing is to certify the presence of a probabilistic mixture of the reference experiment and its complex conjugate. Due to this deficiency and the fact that complex conjugation is not a physical operation, only purely real reference experiments can be properly self-tested. In the introduction we gave an overview of the known results in self-testing and indeed all the states and measurements which allow for self-testing have a purely real representation (\cite{MKYS}-\cite{Bancal},\cite{Ivan}). In Ref. \cite{MkM} the authors deal more rigorously with the problem and even show that for some cryptographic purposes self-testing of the reference experiment involving complex measurements does not undermine security. We note in Appendix \ref{app1} that for our work we may not need to restrict to purely real reference experiments: an assemblage is not typically invariant under taking the complex conjugate of both the state and measurements. For simplicity we will study experiments with states and measurements that have real coefficients but note that an advantage of basing self-testing on EPR-steering eliminates the restriction to only real coefficients.

However, for an arbitrary physical experiment there may exist operations not included in the list above that leave the assemblage and reduced state unchanged. The essence of self-testing based on an assemblage and reduced state is to establish that the only operations a provider can perform that leave it unchanged are those described above.

\subsection{Reduced states and the purification principle}

Given our formalism, the self-testing of quantum states is rendered extremely easy due to the purification principle: every density matrix $\rho_{A}$ on some system $A$ can result as the marginal state of some bipartite pure state $\ket{\psi}_{AB}$ on the joint system $AB$ such that $\rho_{A}=\textrm{tr}(\ket{\psi}_{AB}\bra{\psi}_{AB})$, and this pure state is uniquely defined up to an isometry on system $B$. Therefore, in our formalism, we can observe that given a reduced state $\rho_{C}=\textrm{tr}_{P}(\ket{\psi}\bra{\psi})$ we can describe the state $\ket{\psi}$ upto an isometry on provider's system. In particular, due to the Schmidt decomposition of the reduced state $\rho_{C}=\sum_{i}\lambda_{i}\ket{\mu_{i}}\bra{\mu_{i}}$ (such that $\sum_{i}\lambda_{i}=1$ and $\lambda_{i}\geq0$ for all $i$) we have a purification of the form:
\begin{equation*}
\ket{\psi}=\sum_{i}\sqrt{\lambda_{i}}\ket{\mu_{i}}\ket{\nu_{i}}
\end{equation*}
where $\{\ket{\mu_{i}}\}_{i}$ ($\{\ket{\nu_{i}}\}_{i}$) is some set of orthogonal states in $\mathcal{H}_{C}$ ($\mathcal{H}_{P}$). The local isometry $\Phi:\mathcal{H}_{P}\rightarrow\mathcal{H}_{P}\otimes\mathcal{H}'_{P}$ then maps the set ($\{\ket{\nu_{i}}\}_{i}$) to another set of orthogonal states ($\{\ket{\nu'_{i}}\}_{i}$).

As a consequence of our formalism, we can establish that $\ket{\tilde{\psi}}$ and $\ket{\psi}$ are equivalent solely by checking to see if the reduced state $\tilde{\rho}_{C}=\textrm{tr}_{P}(\ket{\tilde{\psi}}\bra{\tilde{\psi}})$ is equal to the reduced state $\rho_{C}=\textrm{tr}_{P}(\ket{\psi}\bra{\psi})$. Another obvious consequence for entanglement verification between the client and provider is that they share some entanglement if and only if $\rho_{C}$ is mixed. This is purely a consequence of the assumption that they share a pure state. Indeed, it is cryptographically well-motivated to say that the provider produces a pure state since this gives the provider \textit{maximal information} about the devices that are used in a protocol.

Even though self-testing of states is rendered easy by our assumptions, the self-testing of measurements does not follow from only looking at the reduced state $\tilde{\rho}_{C}$. In other words, knowing the global pure $\ket{\psi}$ from the reduced state $\tilde{\rho}_{C}$, does not immediately imply that the provider is making the required measurements on a useful part of that pure state. It should be emphasized that in any one-sided device-independent quantum information protocol, measurements will be made on a state in any task to extract classical information from the systems, both trusted and untrusted. The self-testing of measurements made by an untrusted agent is, as explicitly stated in Eq. \eqref{equiv}, crucial. We give a simple example to illustrate this point. This is an example of a physical system that a provider can prepare and a measurement they can perform.

\begin{Example}
Establishing that the client and provider share a state that is equivalent to a reference state is not immediately useful. Consider the situation where the provider prepares the state $\ket{\psi'}=\frac{1}{\sqrt{2}}\left(\ket{0_{C}}\ket{0_{P_{1}}}\ket{0_{P_{2}}}+\ket{1_{C}}\ket{1_{P_{1}}}\ket{0_{P_{2}}}\right)$ where the subscripts $P_{1}$ and $P_{2}$ label two qubits that the provider retains and sends the qubit with the subscript $C$ to the client. The two qubits labelled by $P_{1}$ and $P_{2}$ can be jointly measured or individually measured. 
In this example the provider's measurement solely consists of measuring qubit $P_{2}$ and ignoring qubit $P_{1}$ such that measurement projectors are of the form $\mathbb{I}_{P_{1}}\otimes (E_{a|x})_{P_{2}}$. Therefore, the reduced state of the client is $\rho_{C}=\frac{\mathbb{I}}{2}$ which indicates that the client and provider share a maximally entangled state. However, every element of the assemblage $\{\sigma_{a|x}\}_{a,x}$ is $\sigma_{a|x}=\frac{\mathbb{I}}{2}$, and thus unaffected by any measurement performed by the provider. Therefore we cannot say anything about the provider's measurements and, furthermore, the entanglement is not being utilised by the provider and will thus not be useful for any quantum information task. 
\end{Example}
This example just highlights that in our scenario it only makes sense to establish equivalence between a physical experiment and reference experiment taking into account \textit{both the state and measurements}. The example motivates the need to study the assemblage generated in our scenario and not just the reduced state. Also, as will be shown later, this allows us to construct explicit isometries demonstrating equivalence between a physical and reference experiment instead of just knowing that such an isometry exists. In colloquial terms, being able to explicitly construct an isometry allows one to be able to ``locate" their desired state within the physical state.

So far we have assumed perfect equivalence between the reference and physical experiment as described by Eqs. \eqref{equiv}. In Sec. \ref{secrob} we extend our discussion to the case where equivalence can be established approximately which is known as robust self-testing. Instead of using the reduced state of the client and assemblage, we may wish to study self-testing given the correlations resulting from measurements on the assemblage and we discuss this in Sec. \ref{seccor}.


\subsection{Robust assemblage-based one-sided self-testing}\label{secrob}

In this section we formally introduce \textit{robust assemblage-based one-sided self-testing} (AST) and indicate its advantages and limitations. Before this we need to recall some mathematical notation in order to discuss ``robustness". We need an appropriate distance measure between operators acting on a Hilbert space. To facilitate this we will use the Schatten $1$-norm $\Vert A\Vert_{1}$ for $A\in\mathcal{L}(\mathcal{H})$ being a linear operator acting on $\mathcal{H}$. This norm is directly related to $D(\rho,\sigma)$, the \textit{trace distance} between quantum states since $D(\rho,\sigma)=\frac{1}{2}\Vert\rho-\sigma\Vert_{1}$ for $\rho$, $\sigma\in\mathcal{D}(\mathcal{H})$. Equivalently, $D(\rho,\sigma)=\frac{1}{2}\sum_{i}\vert\lambda_{i}\vert$ where $\lambda_{i}$ is the $i$th eigenvalue of the operator $(\rho-\sigma)$. Another property of the trace distance is that when $\rho=\ket{a}\bra{a}$ and $\sigma=\ket{b}\bra{b}$ are pure then $D(\ket{a}\bra{a},\ket{b}\bra{b})=\sqrt{1-\vert\langle a\ket{b}\vert^2}$ \cite{NC}. 

The motivation for introducing a distance measure is clear when we consider imperfect experiments. That is, if our physical experiment deviates from the predictions of our reference experiment by a small amount can we be sure that our physical experiment is (up to a local isometry on $\mathcal{H}_{P}$) close (in the trace distance) to our reference experiment? Now we can utilise the trace distance to describe closeness between the physical state $\ket{\psi}$ and reference state $\ket{\tilde{\psi}}$. To whit, if $D(\rho_{C},\tilde{\rho}_{C})=\epsilon>0$ where $\tilde{\rho}_{C}=\textrm{tr}_{P}(\ket{\tilde{\psi}}\bra{\tilde{\psi}})$ and allowing for isometries $\Phi$ on the provider's side, then the minimal distance between physical and reference states will be the minimal value of 
\begin{equation}\label{distance}
D\left(\ket{\Phi}\bra{\Phi},\ket{\mathcal{A}}\bra{\mathcal{A}}\otimes\ket{\tilde{\psi}}\bra{
\tilde{\psi}}\right)=\sqrt{1-\vert\bra{\mathcal{A}}\langle{\tilde{\psi}}\ket{\Phi}\vert^2}
\end{equation}
for $\ket{\Phi}=\Phi(\ket{\psi})$. Clearly, $D(\ket{\Phi}\bra{\Phi},\ket{\mathcal{A}}\bra{\mathcal{A}}\otimes\ket{\tilde{\psi}}\bra{\tilde{\psi}})\geq D(\tilde{\rho}_{C},\rho_C)=\epsilon$ since the trace distance does not increase when tracing out the provider's sub-system.

This lower bound on the distance in Eq. \ref{distance} does not tell us that there is an isometry achieving this bound. We wish to be able to state that there exists an isometry for which the distance in Eq. \eqref{distance} is small. Furthermore it would be preferable to be able to construct this isometry. This is, in essence, robust self-testing. We now formalise this intuition in the following definition:

\begin{definition}
Given a reference experiment consisting of the state $\ket{\tilde{\psi}}\in\mathcal{H}_{C}\otimes\mathcal{H}'_{P}$ with reduced state $\tilde{\rho}_{C}$ and measurements $\{\tilde{E}_{a|x}\}_{a,x}$ such that the assemblage $\{\tilde{\sigma}_{a|x}\}_{a,x}$ has elements $\tilde{\sigma}_{a|x}=\textrm{tr}_{P}\left(\mathbb{I}_{C}\otimes \tilde{E}_{a|x}\ket{\tilde{\psi}}\right)$, $\forall$ $a$, $x$. Also given a physical experiment with the state $\ket{\psi}\in\mathcal{H}_{C}\otimes\mathcal{H}_{P}$, reduced state $\rho_{C}$ and measurements $\{E_{a|x}\}_{a,x}$ such that the assemblage $\{\sigma_{a|x}\}_{a,x}$ has elements $\sigma_{a|x}=\textrm{tr}_{P}\left(\mathbb{I}_{C}\otimes E_{a|x}\ket{\psi}\right)$, $\forall$ $a$, $x$. If, for some real $\epsilon>0$, $D(\tilde{\rho}_{C},\rho_C)\leq\epsilon$ and $\Vert \tilde{\sigma}_{a|x}-\sigma_{a|x}\Vert_{1}\leq\epsilon$, $\forall$ $a$, $x$, then $f(\epsilon)$\textbf{-robust assemblage-based one-sided self-testing ($f(\epsilon)$-AST) is possible} if the assemblage $\{\sigma_{a|x}\}_{a,x}$ implies that there exists an isometry $\Phi:\mathcal{H}_{P}\rightarrow\mathcal{H}_{P}\otimes\mathcal{H}'_{P}$ such that
\begin{align}
\label{eq:defAST}
D\left(\ket{\Phi}\bra{\Phi},\ket{\mathcal{A}}\bra{\mathcal{A}}\otimes\ket{\tilde{\psi}}\bra{\tilde{\psi}}\right)&\leq f(\epsilon),\nonumber\\
\Vert\ket{\Phi,E_{a|x}}\bra{\Phi,E_{a|x}}-\ket{\mathcal{A}}\bra{\mathcal{A}}\otimes(\mathbb{I}_{C}\otimes \tilde{E}_{a|x})\ket{\tilde{\psi}}\bra{\tilde{\psi}}(\mathbb{I}_{C}\otimes \tilde{E}_{a|x})\Vert_{1}&\leq f(\epsilon)
\end{align}
for $\ket{\Phi}=\Phi(\ket{\psi})$, $\ket{\Phi,E_{a|x}}=\Phi(\mathbb{I}_{C}\otimes E_{a|x}\ket{\psi})$, $\ket{\mathcal{A}}\in\mathcal{H}_{P}$ and $f:\mathbb{R}\rightarrow\mathbb{R}$.
\end{definition}

In this definition, in order to simplify matters, we have bounded both the distance between physical and reference states both with and without measurements by the function $f(\epsilon)$. It will often be the case that the trace distance between states (without measurements) will be smaller than the distance between measured states, but we are considering the \textit{worst case} analysis. In further study, it could be of interest to give a finer distinction between these distance measures in the definition.

Note also that, in this definition, we only ask for the existence of an isometry. Later, in Sec. \ref{sec2}, we will construct an isometry for robust self-testing which will be more useful for various protocols. Also, for this definition to be useful, a desirable function would be $f(\epsilon)\leq O(\epsilon^{\frac{1}{p}})$ where $p$ is upper-bounded by a small positive integer. If $D(\tilde{\rho}_{C},\rho_C)=\epsilon$, as mentioned earlier this establishes a lower bound on the distance between physical and reference experiments, and so the ideal case would be $O(\epsilon)$-AST. We now give a simple example to show that, in general, this ideal case is not obtainable.
\begin{Example}
The client has a three-dimensional Hilbert space $\mathcal{H}_{C}$. The reference experiment consists of the state $\ket{\tilde{\psi}}=\frac{1}{\sqrt{2}}\left(\ket{0_{C}0_{P}}+\ket{1_{C}1_{P}}\right)$ with measurements $\{\tilde{E}_{0|0}=\ket{0_{P}}\bra{0_{P}},\tilde{E}_{1|0}=\ket{1_{P}}\bra{1_{P}},\tilde{E}_{0|1}=\ket{+_{P}}\bra{+_{P}},\tilde{E}_{1|1}=\ket{-_{P}}\bra{-_{P}}\}$ and $\ket{\pm_{P}}=\frac{1}{\sqrt{2}}\left(\ket{0_P}\pm\ket{1_P}\right)$ where $\mathcal{H}'_{P}$ is a two-dimensional Hilbert space. The assemblage for this reference experiment has the following elements:
\begin{align}
\tilde{\sigma}_{0|0}&=\frac{1}{2}\ket{0_{C}}\bra{0_{C}},&\tilde{\sigma}_{1|0}&=\frac{1}{2}\ket{1_{C}}\bra{1_{C}},\nonumber\\
\tilde{\sigma}_{0|1}&=\frac{1}{2}\ket{+_{C}}\bra{+_{C}},&\tilde{\sigma}_{1|1}&=\frac{1}{2}\ket{-_{C}}\bra{-_{C}}.\nonumber
\end{align}
The physical experiment consists of the state $\ket{\psi}=\sqrt{1-\epsilon}\ket{\tilde{\psi}}\ket{0_{P'}}+\sqrt{\epsilon}\ket{\xi}\ket{1_{P'}}$ where $\ket{\xi}=\ket{2_{C}0_{P}}$ and the subscript $P'$ denotes a second qubit that the provider has in their possession. The measurements in the physical experiment are $E_{i|j}=\tilde{E}_{i|j}\otimes\ket{0_{P'}}\bra{0_{P'}}+\ket{i_{P}}\bra{i_{P}}\otimes\ket{1_{P'}}\bra{1_{P'}}$ for $i\in\{0,1\}$. The state $\ket{\psi}$ has the reduced state $\rho_{C}=\frac{(1-\epsilon)}{2}(\ket{0_{C}}\bra{0_{C}}+\ket{1_{C}}\bra{1_{C}})+\epsilon
\ket{2_{C}}\bra{2_{C}}$ thus implying that $D(\rho_{C},\tilde{\rho}_{C})=\epsilon$. The assemblage for this physical experiment then has the elements:
\begin{align}
\sigma_{0|0}&=\frac{(1-\epsilon)}{2}\ket{0_{C}}\bra{0_{C}}+\epsilon\ket{2_{C}}\bra{2_{C}},&
\sigma_{1|0}&=\frac{(1-\epsilon)}{2}\ket{1_{C}}\bra{1_{C}},\nonumber\\
\sigma_{0|1}&=\frac{(1-\epsilon)}{2}\ket{+_{C}}\bra{+_{C}}+\epsilon\ket{2_{C}}\bra{2_{C}},&
\sigma_{1|1}&=\frac{(1-\epsilon)}{2}\ket{-_{C}}\bra{-_{C}}.\nonumber
\end{align}
From the above assemblages we observe that $\Vert \tilde{\sigma}_{a|x}-\sigma_{a|x}\Vert_{1}<\frac{3}{2}\epsilon=\epsilon'$, $\forall$ $a$, $x$. Here we have just defined a new closeness parameter $\epsilon'$ for the convenience of our definitions.  Given these physical and reference experiments, we now wish to calculate a lower bound on $D\left(\ket{\Phi}\bra{\Phi},\ket{\mathcal{A}}\bra{\mathcal{A}}\otimes\ket{\tilde{\psi}}\bra{\tilde{\psi}}\right)$ for all possible isometries $\Phi$ in the definition above; this will give a lower-bound on the function $f(\epsilon')$ for $f(\epsilon')$-AST. To do this, we introduce the notation $\ket{\tilde{0}}$ for the ancillae that the provider can introduce and $U_{P}$ as the unitary that they can perform jointly on the ancillae and their share of the physical state $\ket{\psi}$. This then gives us:
\begin{equation*}
D\left(U_{P}\left(\ket{\psi}\bra{\psi}\otimes\ket{\hat{0}}\bra{\hat{0}}\right)U_{P}^{\dagger},\ket{\mathcal{A}}\bra{\mathcal{A}}\otimes\ket{\tilde{\psi}}\bra{\tilde{\psi}}\right)=\sqrt{1-F^{2}},
\end{equation*}
where
\begin{align}
F&=\vert\bra{\mathcal{A}}\bra{\tilde{\psi}}U\ket{\psi}\ket{\hat{0}}\vert\nonumber\\
&=\frac{\sqrt{1-\epsilon}}{2}\vert\bra{\mathcal{A}}\left(\bra{0_{C}0_{P}}(\mathbb{I}_{C}\otimes U_{P})\ket{0_{C}0_{P}}+\bra{1_{C}1_{P}}(\mathbb{I}_{C}\otimes U_{P})\ket{1_{C}1_{P}}\right)\ket{\hat{0}}\vert,\nonumber
\end{align}
where $\mathbb{I}_{C}$ is the identity on the client's system. Thus maximizing this quantity for all isometries, we obtain the maximal value $F^{*}=\sqrt{1-\epsilon}=\sqrt{1-\frac{2\epsilon'}{3}}$ and the lower bound $D\left(\ket{\Phi}\bra{\Phi},\ket{\mathcal{A}}\bra{\mathcal{A}}\otimes\ket{\tilde{\psi}}\bra{\tilde{\psi}}\right)\geq\sqrt{\frac{2\epsilon'}{3}}$. 
\end{Example}
This example excludes the possibility of having $O(\epsilon)$-AST given that the client's Hilbert space is three-dimensional. We will later return to this reference experiment in Sec. \ref{sec2a} with the modification that the client's Hilbert space is two-dimensional.

\subsection{Robust correlation-based one-sided self-testing}\label{seccor}

As outlined earlier, EPR-steering can be studied from the point-of-view of the probabilities obtained from measurements performed on elements of an assemblage, i.e. known measurements made by the trusted party. This point-of-view is native to Bell non-locality and is suitable for making further parallels between non-locality and EPR-steering. In this regard one can construct  
EPR-steering inequalities (the EPR-steering analogues of Bell inequalities) which can be written as a linear combination of the measurement probabilities \cite{CWJR}. The two figures-of-merit, assemblages and measurement correlations, lead to a certain duality in the theory of EPR steering. The approach that one will use depends on the underlying scenario. In the case when correlations are obtained by performing a tomographically complete set of measurements (on the trusted system) the two approaches become completely equivalent. However, in some cases probabilities obtained by performing a tomographically incomplete set of measurements, or even just the amount of violation of some steering inequality can provide all necessary information. Another possibility is that a trusted party can perform only two measurements and nothing more, i.e. has no resources to perform complete tomography. In this section we consider the definition and utility of defining robust self-testing with respect to these probabilities for an appropriate notion of robustness. This approach to self-testing is not immediately equivalent to the notion of AST defined previously (even if tomographically complete measurements are made) for reasons that will be become clear.

Recall the probabilities $p(a,b|x,y)=\textrm{tr}(F_{b|y}\sigma_{a|x})$ for $F_{b|y}$ being elements of general measurement associated with the outcome $b$ for measurement choice $y$ such that $\sum_{b}F_{b|y}=\mathbb{I}_{C}$. Naturally, we can also obtain the probabilities $p(b|y)=\textrm{tr}(F_{b|y}\rho_{C})$. In addition to the ``physical probabilities" $p(a,b|x,y)$, we have the ``reference probabilities" $\{\tilde{p}(a,b|x,y)\}$ which refer to the probabilities resulting from making the same measurements $\{F_{b|y}\}_{b,y}$ on a reference assemblage $\{\tilde{\sigma}_{a|x}\}$ as described above. Performing robust self-testing given these probabilities will be the focus of this section.

A useful definition of the Schatten $1$-norm is $\Vert A\Vert_{1}={\textrm{sup}}_{\Vert B\Vert\leq 1}\vert\textrm{tr}(BA)\vert$ where $\Vert\cdot\Vert$ is the operator norm. Since $F_{b|y}$ is a positive operator with operator norm upper bounded by $1$ and if $D(\rho_{C},\tilde{\rho}_{C})\leq \epsilon$ and for all elements $\sigma_{a|x}$ of an assemblage $\Vert \tilde{\sigma}_{a|x}-\sigma_{a|x}\Vert_{1}\leq\epsilon$ we can conclude that
\begin{align}
\vert \tilde{p}(a,b|x,y)-p(a,b|x,y)\vert&=\vert \textrm{tr}\left[F_{b|y}\left(\sigma_{a|x}-\tilde{\sigma}_{a|x}\right)\right]\vert\leq\Vert \tilde{\sigma}_{a|x}-\sigma_{a|x}\Vert_{1}\leq\epsilon,\nonumber\\
\vert p(b|y)-\tilde{p}(b|y)\vert&=\vert \textrm{tr}\left[F_{b|y}\left(\rho_{C}-\tilde{\rho}_{C}\right)\right]\vert\leq2D(\rho_{C},\tilde{\rho}_{C})\leq 2\epsilon\nonumber
\end{align}
for all $a$, $b$, $x$, $y$. This then establishes that knowledge of the assemblage and establishing its closeness to the assemblage associated with a reference experiment implies closeness in the probabilities obtained from both experiments. Clearly, the converse is not necessarily true and closeness in probabilities does not always imply closeness of reduced states and assemblages. Assemblages can be calculated from the statistics obtained from performing tomographically complete measurements, and then the distance (in Schatten $1$-norm) between this assemblage and some ideal assemblage can be calculated. However, even for tomographically complete measurements $\{F_{b|y}\}_{b,y}$, we only have that $\vert \textrm{tr}\left[F_{b|y}\left(\sigma_{a|x}-\tilde{\sigma}_{a|x}\right)\right]\vert\leq\Vert \tilde{\sigma}_{a|x}-\sigma_{a|x}\Vert_{1}$ thus having $\vert \textrm{tr}\left[F_{b|y}\left(\sigma_{a|x}-\tilde{\sigma}_{a|x}\right)\right]\vert\leq \epsilon$ does not imply $\Vert \tilde{\sigma}_{a|x}-\sigma_{a|x}\Vert_{1}\leq\epsilon$. This goes to show that the AST approach is distinct from solely looking at the difference between probabilities.


Inspired by the literature in standard self-testing (see, e.g. Refs. \cite{MKYS,BQC}), it should still be possible to attain robust self-testing based on probabilities for measurements on assemblages and with this in mind, we give the following definition:

\begin{definition}
Given a reference experiment consisting of the state $\ket{\tilde{\psi}}\in\mathcal{H}_{C}\otimes\mathcal{H}'_{P}$ with reduced state $\tilde{\rho}_{C}$ and measurements $\{\tilde{E}_{a|x}\}_{a,x}$ such that the assemblage $\{\tilde{\sigma}_{a|x}\}_{a,x}$ has elements $\tilde{\sigma}_{a|x}=\textrm{tr}_{P}\left(\mathbb{I}_{C}\otimes \tilde{E}_{a|x}\ket{\tilde{\psi}}\right)$, $\forall$ $a$, $x$. Also given a physical experiment with the state $\ket{\psi}\in\mathcal{H}_{C}\otimes\mathcal{H}_{P}$, reduced state $\rho_{C}$ and measurements $\{E_{a|x}\}_{a,x}$ such that the assemblage $\{\sigma_{a|x}\}_{a,x}$ has elements $\sigma_{a|x}=\textrm{tr}_{P}\left(\mathbb{I}_{C}\otimes E_{a|x}\ket{\psi}\right)$, $\forall$ $a$, $x$. Additionally given a set $\{F_{b|y}\}_{b,y}$ of general measurements that act on $\mathcal{H}_{C}$ such that $p(a,b|x,y)=\textrm{tr}\left(F_{b|y}\sigma_{a|x}\right)$ and $\tilde{p}(a,b|x,y)=\textrm{tr}\left(F_{b|y}\tilde{\sigma}_{a|x}\right)$ $\forall$ $a$, $x$. If, for some real $\epsilon>0$,
\begin{align}
\vert \tilde{p}(a,b|x,y)-p(a,b|x,y)\vert\leq\epsilon,\nonumber\\
\vert \tilde{p}(b|y)-p(b|y)\vert\leq\epsilon,\nonumber\\
\vert \tilde{p}(a|x)-p(a|x)\vert\leq\epsilon,\nonumber
\end{align}
 $\forall$ $a$, $x$, $b$, $y$, then $f(\epsilon)$\textbf{-robust correlation-based one-sided self-testing ($f(\epsilon)$-CST)} is possible if the probabilities imply that there exists an isometry $\Phi:\mathcal{H}_{P}\rightarrow\mathcal{H}_{P}\otimes\mathcal{H}'_{P}$ such that
\begin{align}
D\left(\ket{\Phi}\bra{\Phi},\ket{\mathcal{A}}\ket{\tilde{\psi}}\bra{\mathcal{A}}\bra{\tilde{\psi}}\right)&\leq f(\epsilon),\nonumber\\
\Vert\ket{\Phi,E_{a|x}}\bra{\Phi,E_{a|x}}-\ket{\mathcal{A}}(\mathbb{I}_{C}\otimes \tilde{E}_{a|x})\ket{\tilde{\psi}}\bra{\mathcal{A}}\bra{\tilde{\psi}}(\mathbb{I}_{C}\otimes \tilde{E}_{a|x})\Vert_{1}&\leq f(\epsilon)\nonumber
\end{align}
for $\ket{\Phi}=\Phi(\ket{\psi})$, $\ket{\Phi,E_{a|x}}=\Phi(\mathbb{I}_{C}\otimes E_{a|x}\ket{\psi})$, $\ket{\mathcal{A}}\in\mathcal{H}_{P}$ and $f:\mathbb{R}\rightarrow\mathbb{R}$.
\end{definition}

Instead of directly bounding the distance between reference and physical probabilities, we can indirectly bound this distance by utilising an EPR-steering inequality. In the literature on standard self-testing, probability distributions that near-maximally violate a Bell inequality robustly self-test the state and measurements that produce the maximal violation \cite{MKYS,BQC}. As a first requirement, there needs to be a unique probability distribution that achieves this maximal violation, and we now have many examples of Bell inequalities where this happens. The same applies to EPR-steering inequalities: there needs to be a unique assemblage that produces the maximal violation of an EPR-steering inequality. Furthermore this unique assemblage needs to imply a unique reference experiment (up to a local isometry). For EPR-steering inequalities of the form $\sum_{a|x}\alpha_{a,x}\textrm{tr}\left(F_{a|b}\sigma_{a|x}\right)\geq 0$ for real numbers $\alpha_{a,x}$, any assemblage that violates this inequality is necessarily \textit{steerable}. If all quantum assemblages satisfy $\sum_{a|x}\alpha_{a,x}\textrm{tr}\left(F_{a|b}\sigma_{a|x}\right)\geq -\beta$ for some positive real number $\beta$ then $-\beta$ is the maximal violation of the EPR-steering inequality. If we consider probabilities of the form $p(a,b|x,y)=\textrm{tr}\left(F_{b|y}\sigma_{a|x}\right)$ that satisfy $\sum_{a|x}\alpha_{a,x}\textrm{tr}\left(F_{a|b}\sigma_{a|x}\right)\leq-(\beta-\epsilon)$ then they are at most $\epsilon$-far from the reference experiment that produces the maximal violation of $-\beta$. We will make use of this approach to CST in Sec. \ref{sec2b}.

We now briefly return to the issue of complex conjugation. As mentioned above and discussed in Appendix \ref{app1}, the AST approach is advantageous to the standard self-testing approach in that we can rule out the state and measurements in the reference experiment both being the complex conjugate of our ideal reference experiment. One issue with CST is that since we are reconsidering probabilities for a fixed set of measurements made by the client, if the measurements are invariant under complex conjugation then the provider can prepare a state and make measurements that are both the complex conjugate of the ideal case without altering the statistics. This can be remedied by the client choosing measurements that have complex entries as long as it does not drastically affect the ability to achieve $f(\epsilon)$-CST.

\section{Self-testing of an ebit}\label{sec2}

In this section, we look at the self-testing of the maximally entangled two-qubit state (or, \textit{ebit}). This is a totemic state in the self-testing literature (e.g. \cite{MKYS,BQC}) and that it is possible to do RST for this state is now well-established: it is achieved by looking at probability distributions that near-maximally violate the CHSH inequality. That is, since the maximal violation of the CHSH inequality is, say, $2\sqrt{2}$ then probability distributions that give a violation of $2\sqrt{2}-\epsilon$ result from quantum states that are $O(\sqrt{\epsilon})$-close to the ebit (up to local isometries). In current analytical approaches the constant in front of the $\sqrt{\epsilon}$ term can be shown to be quite large. However, there are numerical approaches that substantially improve upon this constant by several orders of magnitude \cite{YVBSN,Bancal}. 

We turn to AST and CST to see if we can improve the current approaches that appear for RST. In particular, in Sec. \ref{sec2a} we look at analytical methods for AST and show that, for the ebit, $O(\sqrt{\epsilon})$-AST is possible where the constant in front of the $\sqrt{\epsilon}$ term is reasonable. In Sec. \ref{sec2b} we turn to numerical methods for CST where the study of probabilities instead of assemblages is currently more amenable. We show that $O(\sqrt{\epsilon})$-CST is possible and also that our numerical methods do better than existing numerical methods for RST. Thirdly, in Sec. \ref{sec2c} we then show that $O(\sqrt{\epsilon})$-AST is essentially the best that one can hope for by explicitly giving a physical state and measurements where $f(\epsilon)$ in the definition of $f(\epsilon)$-AST will be at least $\sqrt{\epsilon}$. In other words, $O(\epsilon)$-AST is impossible.

\subsection{Analytical results utilising the SWAP isometry}\label{sec2a}

We first set-out the reference experiment that we will be studying for the rest of this section. It consists of the experiment described in Sec. \ref{secrob} but now with the client's Hilbert space being two-dimensional. Recall that the state is $\ket{\tilde{\psi}}=\frac{1}{\sqrt{2}}\left(\ket{00}+\ket{11}\right)$ and the measurements are $\{\tilde{E}_{0|0}=\ket{0}\bra{0},\tilde{E}_{1|0}=\ket{1}\bra{1},\tilde{E}_{0|1}=\ket{+}\bra{+},\tilde{E}_{1|1}=\ket{-}\bra{-}\}$ and $\ket{\pm}=\frac{1}{\sqrt{2}}\left(\ket{0}\pm\ket{1}\right)$ where we have dropped the subscripts for reasons of clarity. The assemblage for this reference experiment has the following elements:
\begin{align}
\tilde{\sigma}_{0|0}&=\frac{1}{2}\ket{0}\bra{0},&\tilde{\sigma}_{1|0}&=\frac{1}{2}\ket{1}\bra{1},\nonumber\\
\tilde{\sigma}_{0|1}&=\frac{1}{2}\ket{+}\bra{+},&\tilde{\sigma}_{1|1}&=\frac{1}{2}\ket{-}\bra{-}.\nonumber
\end{align}
We will henceforth call this reference experiment the \textit{EPR experiment}. We can now state a result about AST for this experiment.

\begin{theorem}\label{thm1}
For the EPR experiment, $f(\epsilon)$-robust assemblage-based one-sided self-testing is possible for $f(\epsilon)=24\sqrt{\epsilon}+\epsilon$.
\end{theorem}

Before proving this theorem we will present two useful observations that will be used in the proof. The first observation is a lemma about the norm that we are using while the second is specific to the self-testing of the EPR experiment. We require the notation $\Vert\ket{v}\Vert=\sqrt{\langle v\ket{v}}$. 

\begin{lemma}\label{goodlem}
For any two vectors $\ket{u}$, $\ket{v}$ where $\Vert{\ket{u}}\Vert\leq 1$ and $\Vert{\ket{v}}\Vert\leq 1$, if $\Vert\ket{u}-\ket{v}\Vert\leq\eta\leq 1$, then for another vector $\ket{t}$ such that $\Vert{\ket{t}}\Vert\leq \beta$, $\Vert{\left(\ket{u}-\ket{v}\right)}\bra{t}\Vert_{1}\leq\beta\eta$ and $\Vert\ket{t}{\left(\bra{u}-\bra{v}\right)}\Vert_{1}\leq\beta\eta$
\end{lemma}

\begin{proof}
This fact essentially follows from the definition of $\Vert\cdot\Vert$. That is, $\Vert\ket{u}-\ket{v}\Vert=\sqrt{\langle{u}\ket{u}+\langle{v}\ket{v}-\langle{u}\ket{v}-\langle{v}\ket{u}}$ and since the rank of $B=\left(\ket{u}-\ket{v}\right)\bra{t}$ is $1$ then the $\Vert{B}\Vert_{1}=\sqrt{\textrm{tr}\left(BB^{\dagger}\right)}=\Vert\ket{t}\Vert\sqrt{\langle{u}\ket{u}+\langle{v}\ket{v}-\langle{u}\ket{v}-\langle{v}\ket{u}}$ which concludes our proof (along with the fact that $\Vert B\Vert_{1}=\Vert B^{\dagger}\Vert_{1}$).
\end{proof}

The next observation follows from the conditions outlined in the definition of $f(\epsilon)$-AST and is as follows:

\begin{lemma}\label{niceobs}
If $\Vert\sigma_{a|x}-\tilde{\sigma}_{a|x}\Vert_{1}\leq\epsilon$ and $D(\rho_{C},\tilde{\rho}_{C})\leq\epsilon$ then
\begin{equation*}
\Vert\mathbb{I}_{C}\otimes E_{a|x}\ket{\psi}-\tilde{E}_{a|x}\otimes\mathbb{I}_{P}\ket{\psi}\Vert\leq 2\sqrt{\epsilon}
\end{equation*}
\end{lemma}

\begin{proof}
The proof follows from a series of basic observations:
\begin{align}
\Vert\mathbb{I}_{C}\otimes E_{a|x}\ket{\psi}-\tilde{E}_{a|x}\otimes\mathbb{I}_{P}\ket{\psi}\Vert&=\sqrt{\bra{\psi}\mathbb{I}_{C}\otimes E_{a|x}\ket{\psi}+\bra{\psi}\tilde{E}_{a|x}\otimes\mathbb{I}_{P}\ket{\psi}-2\bra{\psi}\tilde{E}_{a|x}\otimes E_{a|x}\ket{\psi}}\nonumber\\
&\leq\sqrt{1+2\epsilon-2\bra{\psi}\tilde{E}_{a|x}\otimes E_{a|x}\ket{\psi}}\nonumber\\
&\leq\sqrt{1+2\epsilon-2(\frac{1}{2}-\epsilon)}\nonumber\\
&=2\sqrt{\epsilon}.\nonumber
\end{align}
The first inequality results from the fact that $\bra{\psi}\mathbb{I}_{C}\otimes E_{a|x}\ket{\psi}=\textrm{tr}_{P}\left(\sigma_{a|x}\right)$ and $\bra{\psi}\tilde{E}_{a|x}\otimes\mathbb{I}_{P}\ket{\psi}=\textrm{tr}_{P}(\tilde{E}_{a|x}\rho_{C})$ and that $\vert\textrm{tr}\left(\sigma_{a|x}-\tilde{\sigma}_{a|x}\right)\vert\leq\epsilon$ and $\vert\textrm{tr}(\tilde{E}_{a|x}\rho_{C}-\tilde{E}_{a|x}\tilde{\rho}_{C})\vert\leq\epsilon$. The second inequality follows from the observation that $\vert\textrm{tr}(\tilde{E}_{a|x}\sigma_{a|x}-\tilde{E}_{a|x}\tilde{\sigma}_{a|x})\vert\leq\epsilon$.
\end{proof}

We are now in a position to prove Thm. \ref{thm1}.
\newline

\noindent
\begin{proof} 
Recall that we are promised that
\begin{align}
D(\rho_{C},\tilde{\rho}_{C})&\leq\epsilon,\nonumber\\
\Vert \sigma_{a|x}-\tilde{\sigma}_{a|x}\Vert_{1}&\leq\epsilon,\nonumber
\end{align}
for all $a$, $x$ where $\tilde{\rho}_{C}=\textrm{tr}_{P}\left(\ket{\tilde{\psi}}\bra{\tilde{\psi}}\right)$. The aim is now to find an explicit isometry $\Phi$ that gives a non-trivial upper bound for the following expression:
\begin{equation}\label{condition}
\Vert\ket{\Phi,Q_{a|x}}\bra{\Phi,Q_{a|x}}-\ket{\mathcal{A}}\bra{\mathcal{A}}\otimes(\mathbb{I}_{C}\otimes \tilde{Q}_{a|x})\ket{\tilde{\psi}}\bra{\tilde{\psi}}(\mathbb{I}_{C}\otimes \tilde{Q}_{a|x})\Vert_{1},
\end{equation}
for $Q_{a|x}\in\{\mathbb{I},E_{a|x}\}$, $\tilde{Q}_{a|x}\in\{\mathbb{I},\tilde{E}_{a|x}\}$ and $\ket{\Phi,Q_{a|x}}$ as defined before. We first focus on the cases where $Q_{a|x}=\mathbb{I}_{P}$ and $\tilde{Q}_{a|x}=\mathbb{I}=\mathbb{I}_{C}$ and use this to argue the more general result.

The isometry that we use is the so-called SWAP isometry that has been used multiple times in the self-testing literature. In this isometry (see Fig. \ref{fig:fig2}) an ancilla qubit is introduced in the state $\ket{+_{P'}}\in\mathcal{H}_{P'}$ where $P'$ denotes the ancilla register on the provider's side in addition to the provider's Hilbert space $\mathcal{H}_{P}$. After introducing the ancilla a unitary operator is applied to both the provider's part of the physical state and the ancilla, i.e. $\ket{\psi}\ket{+_{P'}}\rightarrow(\mathbb{I}_{C}\otimes VHU)\ket{\psi}\ket{+_{P'}}$ where $U=\ket{0_{P'}}\bra{0_{P'}}\otimes\mathbb{I}_{P}+\ket{1_{P'}}\bra{1_{P'}}\otimes Z_{P}$, $V=\ket{0_{P'}}\bra{0_{P'}}\otimes\mathbb{I}_{P}+\ket{1_{P'}}\bra{1_{P'}}\otimes X$ and $H=\ket{+_{P'}}\bra{0_{P'}}+\ket{-_{P'}}\bra{1_{P'}}$ and $Z=2E_{0|0}-\mathbb{I}_{P}$, $X=2E_{0|1}-\mathbb{I}_{P}$ and $X^{2}=Z^{2}=\mathbb{I}_{P}$. After applying this isometry to the physical state $\ket{\psi}$ we obtain the state
\begin{equation*}
\ket{\psi'}=E_{0|0}\ket{\psi}\ket{0_{P'}}+XE_{1|0}\ket{\psi}\ket{1_{P'}}.
\end{equation*}

\begin{figure}
  \centering
    \includegraphics[width=0.4\textwidth]{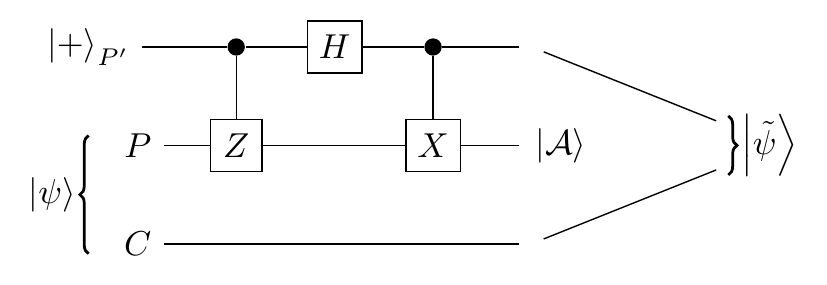}
    \caption{Here the SWAP isometry applied to the provider's system is depicted as a quantum circuit. The notation is explained in the text.\label{fig:fig2}}
\end{figure}

The desired result of this isometry to establish an ebit in the Hilbert space $\mathcal{H}_{C}\otimes\mathcal{H}_{P'}=\mathcal{H}_{C}\otimes\mathcal{H}'_{P}$ in addition to the measurements $\tilde{E}_{a|x}$ acting on the Hilbert space $\mathcal{H}_{P'}$. Therefore we wish to give an upper bound to
\begin{align}\label{tracedist}
&\Vert\left(E_{0|0}\ket{\psi}\ket{0_{P'}}+XE_{1|0}\ket{\psi}\ket{1_{P'}}\right)\left(\bra{\psi}E_{0|0}\bra{0_{P'}}+\bra{\psi}E_{1|0}X\bra{1_{P'}}\right)-\ket{\mathcal{A}}\bra{\mathcal{A}}\otimes\ket{\tilde{\psi}}\bra{\tilde{\psi}}\Vert_{1}.
\end{align}
At this point we can now apply a combination of Lem. \ref{goodlem} and Lem. \ref{niceobs} to bound this norm. Firstly, we observe that by virtue of Lem. \ref{niceobs} we have that
\begin{align}
\Vert\left(E_{0|0}\ket{\psi}\ket{0_{P'}}+XE_{1|0}\ket{\psi}\ket{1_{P'}}\right)-\left(\tilde{E}_{0|0}\otimes\mathbb{I}_{P}\ket{\psi}\ket{0_{P'}}+XE_{1|0}\ket{\psi}\ket{1_{P'}}\right)\Vert &\leq 2\sqrt{\epsilon},\nonumber\\
\Vert\left(\tilde{E}_{0|0}\otimes\mathbb{I}_{P}\ket{\psi}\ket{0_{P'}}+XE_{1|0}\ket{\psi}\ket{1_{P'}}\right)-\left(\tilde{E}_{0|0}\otimes\mathbb{I}_{P}\ket{\psi}\ket{0_{P'}}+\tilde{E}_{1|0}\otimes X\ket{\psi}\ket{1_{P'}}\right)\Vert &\leq 2\sqrt{\epsilon},\nonumber
\end{align}
where, for the sake of brevity, we do not write identities $\mathbb{I}_{C}$, e.g. $E_{0|0}\ket{\psi}\ket{0_{P'}}=\mathbb{I}_{C}\otimes E_{0|0}\ket{\psi}\ket{0_{P'}}$.

We can apply these observations in conjunction with Lem. \ref{goodlem} (and noticing that $\Vert E_{0|0}\ket{\psi}\ket{0_{P'}}+XE_{1|0}\ket{\psi}\ket{1_{P'}}\Vert=1$) to Eq. \ref{tracedist} to obtain
\begin{align}
&\Vert\left(E_{0|0}\ket{\psi}\ket{0_{P'}}+XE_{1|0}\ket{\psi}\ket{1_{P'}}\right)\left(\bra{\psi}E_{0|0}\bra{0_{P'}}+\bra{\psi}E_{1|0}X\bra{1_{P'}}\right)-\ket{\mathcal{A}}\bra{\mathcal{A}}\otimes\ket{\tilde{\psi}}\bra{\tilde{\psi}}\Vert_{1}\nonumber\\
&\leq 2\sqrt{\epsilon}+\Vert\left(\tilde{E}_{0|0}\otimes\mathbb{I}_{P}\ket{\psi}\ket{0_{P'}}+XE_{1|0}\ket{\psi}\ket{1_{P'}}\right)\left(\bra{\psi}E_{0|0}\bra{0_{P'}}+\bra{\psi}E_{1|0}X\bra{1_{P'}}\right)-\ket{\mathcal{A}}\bra{\mathcal{A}}\otimes\ket{\tilde{\psi}}\bra{\tilde{\psi}}\Vert_{1}\nonumber\\
&\leq 4\sqrt{\epsilon}+\Vert\left(\tilde{E}_{0|0}\otimes\mathbb{I}_{P}\ket{\psi}\ket{0_{P'}}+\tilde{E}_{1|0}\otimes X\ket{\psi}\ket{1_{P'}}\right)\left(\bra{\psi}E_{0|0}\bra{0_{P'}}+\bra{\psi}E_{1|0}X\bra{1_{P'}}\right)-\ket{\mathcal{A}}\bra{\mathcal{A}}\otimes\ket{\tilde{\psi}}\bra{\tilde{\psi}}\Vert_{1}\nonumber.
\end{align}
Since $X=2E_{0|1}-\mathbb{I}_{P}$ and, for the Pauli-$X$ matrix $\tau_x=2\ket{+}\bra{+}-\mathbb{I}$, we obtain the following result that
\begin{align}
\Vert\mathbb{I}_{C}\otimes X\ket{\psi}-\tau_x\otimes\mathbb{I}_{P}\ket{\psi}\Vert&\leq 2\Vert\mathbb{I}_{C}\otimes E_{0|1}\ket{\psi}-\tilde{E}_{0|1}\otimes\mathbb{I}_{P}\ket{\psi}\Vert -\Vert\ket{\psi}-\ket{\psi}\Vert\nonumber\\
&\leq 4\sqrt{\epsilon}.\nonumber
\end{align}
We then obtain 
\begin{align}
&\Vert\left(E_{0|0}\ket{\psi}\ket{0_{P'}}+XE_{1|0}\ket{\psi}\ket{1_{P'}}\right)\left(\bra{\psi}E_{0|0}\bra{0_{P'}}+\bra{\psi}E_{1|0}X\bra{1_{P'}}\right)-\ket{\mathcal{A}}\bra{\mathcal{A}}\otimes\ket{\tilde{\psi}}\bra{\tilde{\psi}}\Vert_{1}\nonumber\\
&\leq 8\sqrt{\epsilon}+\Vert\left(\tilde{E}_{0|0}\otimes{\mathbb{I}_{P}}\ket{\psi}\ket{0_{P'}}+\tilde{E}_{1|0}\tau_{x}\otimes{\mathbb{I}_{P}}\ket{\psi}\ket{1_{P'}}\right)\left(\bra{\psi}E_{0|0}\bra{0_{P'}}+\bra{\psi}E_{1|0}X\bra{1_{P'}}\right)-\ket{\mathcal{A}}\bra{\mathcal{A}}\otimes\ket{\tilde{\psi}}\bra{\tilde{\psi}}\Vert_{1}\nonumber.
\end{align}
We will now apply the same reasoning to $\left(\bra{\psi}E_{0|0}\bra{0_{P'}}+\bra{\psi}E_{1|0}X\bra{1_{P'}}\right)$ but we need the fact that 
\begin{equation*}
\Vert \tilde{E}_{0|0}\otimes{\mathbb{I}_{P}}\ket{\psi}\ket{0_{P'}}+\tilde{E}_{1|0}\tau_{x}\otimes{\mathbb{I}_{P}}\ket{\psi}\ket{1_{P'}}\Vert=\sqrt{2\bra{\psi}\tilde{E}_{0|0}\otimes{\mathbb{I}_{P}}\ket{\psi}}\leq\sqrt{1+2\epsilon}\leq 1+\epsilon,
\end{equation*}
which follows from the condition on the reduced state $\rho_{C}$ and $\tilde{E}_{1|0}\tau_{x}=\tau_{x}\tilde{E}_{0|0}$. Using these observations and Lem. \ref{niceobs} we arrive at
\begin{align}
&\Vert\left(E_{0|0}\ket{\psi}\ket{0_{P'}}+XE_{1|0}\ket{\psi}\ket{1_{P'}}\right)\left(\bra{\psi}E_{0|0}\bra{0_{P'}}+\bra{\psi}E_{1|0}X\bra{1_{P'}}\right)-\ket{\mathcal{A}}\bra{\mathcal{A}}\otimes\ket{\tilde{\psi}}\bra{\tilde{\psi}}\Vert_{1}\nonumber\\
&\leq 16\sqrt{\epsilon}+8\epsilon\sqrt{\epsilon}\nonumber\\
&+\Vert\left(\tilde{E}_{0|0}\otimes\mathbb{I}_{P}\ket{\psi}\ket{0_{P'}}+\tilde{E}_{1|0}\tau_{x}\otimes\mathbb{I}_{P}\ket{\psi}\ket{1_{P'}}\right)\left(\bra{\psi}\tilde{E}_{0|0}\otimes\mathbb{I}_{P}\bra{0_{P'}}+\bra{\psi}\tau_{x}\tilde{E}_{1|0}\otimes\mathbb{I}_{P}\bra{1_{P'}}\right)-\ket{\mathcal{A}}\bra{\mathcal{A}}\otimes\ket{\tilde{\psi}}\bra{\tilde{\psi}}\Vert_{1}\nonumber\\
&=16\sqrt{\epsilon}+8\epsilon\sqrt{\epsilon}+\Vert\left(\langle{0_{C}}\ket{\psi}\ket{0_{C}0_{P'}}
+\langle{0_{C}}\ket{\psi}\ket{1_{C}1_{P'}}\right)\left(\bra{\psi}0_{C}\rangle\bra{0_{C}0_{P'}}+
\bra{\psi}0_{C}\rangle\bra{1_{C}1_{P'}}\right)-\ket{\mathcal{A}}\bra{\mathcal{A}}\otimes\ket{\tilde{\psi}}\bra{\tilde{\psi}}\Vert_{1}\nonumber\\
&=16\sqrt{\epsilon}+8\epsilon\sqrt{\epsilon}+\Vert 2\langle{0_{C}}\ket{\psi}\ket{\tilde{\psi}}\bra{\psi}0_{C}\rangle\bra{\tilde{\psi}}
-\ket{\mathcal{A}}\bra{\mathcal{A}}\otimes\ket{\tilde{\psi}}\bra{\tilde{\psi}}\Vert_{1}\nonumber\\
&\leq 16\sqrt{\epsilon}+8\epsilon\sqrt{\epsilon}+\Vert 2\langle{0_{C}}\ket{\psi}\bra{\psi}0_{C}\rangle-\ket{\mathcal{A}}\bra{\mathcal{A}}\Vert_{1}\nonumber\\
&\leq 16\sqrt{\epsilon}+8\epsilon\sqrt{\epsilon}+2\epsilon,\nonumber
\end{align}
where to obtain the last inequality we chose $\ket{\mathcal{A}}$ to be the pure state that is proportional to $\ket{0_{C}}\langle{0_{C}}\ket{\psi}$, i.e. $\ket{\mathcal{A}}=\beta^{-\frac{1}{2}}\ket{0_{C}}\langle{0_{C}}\ket{\psi}$ where $\beta=\bra{\psi}0_{C}\rangle\langle 0_{C}\ket{\psi}$ thus $\vert\textrm{tr}\left(\vert 0_{C}\rangle\langle 0_{C}\vert\rho_{C}\right)-\textrm{tr}\left(\vert 0_{C}\rangle\langle 0_{C}\vert\tilde{\rho}_{C}\right)\vert\leq\vert\beta-\frac{1}{2}\vert\leq\epsilon$.

We have shown that $D(\ket{\Phi}\bra{\Phi},\ket{\mathcal{A}}\bra{\mathcal{A}}\otimes\ket{\tilde{\psi}}\bra{\tilde{\psi}})\leq 8\sqrt{\epsilon}+4\epsilon\sqrt{\epsilon}+\epsilon$. Now we consider the case of self-testing where measurements are made. That is, establishing an upper bound on the expressions of the form in Eq. \ref{condition} where $Q_{a|x}\neq\mathbb{I}_{P}$ and $\tilde{Q}_{a|x}\neq\mathbb{I}$ and after applying the SWAP isometry described above, the projector acting on the physical state $E_{a|x}\ket{\psi}$ gets mapped to
\begin{equation*}
E_{0|0}E_{a|x}\ket{\psi}\ket{0_{P'}}+XE_{1|0}E_{a|x}\ket{\psi}\ket{1_{P'}}.
\end{equation*}
In the case that $x=0$, utilising the fact that $E_{a|x}E_{a'|x}=\delta^{a}_{a'}E_{a|x}$, for Eq. \ref{condition} we obtain:
\begin{align}
\Vert E_{0|0}\ket{\psi}\bra{\psi}E_{0|0}\otimes\ket{0_{P'}}\bra{0_{P'}}-\frac{1}{2}\ket{\mathcal{A}}\bra{\mathcal{A}}\otimes\ket{0_{C}0_{P'}}\bra{0_{C}0_{P'}}\Vert_{1}&\textrm{ for }a=0,\nonumber\\
\Vert XE_{1|0}\ket{\psi}\bra{\psi}E_{1|0}X\otimes\ket{1_{P'}}\bra{1_{P'}}-\frac{1}{2}\ket{\mathcal{A}}\bra{\mathcal{A}}\otimes\ket{1_{C}1_{P'}}\bra{1_{C}1_{P'}}\Vert_{1}&\textrm{ for }a=1\nonumber.
\end{align}
By using the same reasoning as above we obtain the bounds $4\sqrt{\epsilon}+\epsilon$ and $12\sqrt{\epsilon}+\epsilon$ for the $a=0$ and $a=1$ cases respectively. For the case that $x=1$, more work is required in bounding Eq. \ref{condition}. However, again by repeatedly applying the observation in Lem. \ref{niceobs}, as shown in Appendix \ref{app1c} we obtain the bound of 
\begin{equation}\label{final}
\Vert\ket{\Phi,Q_{a|x}}\bra{\Phi,Q_{a|x}}-\ket{\mathcal{A}}\bra{\mathcal{A}}\otimes(\mathbb{I}_{C}\otimes \tilde{Q}_{a|x})\ket{\tilde{\psi}}\bra{\tilde{\psi}}(\mathbb{I}_{C}\otimes \tilde{Q}_{a|x})\Vert_{1}\leq 24\sqrt{\epsilon}+\epsilon,
\end{equation}
thus concluding the proof.
\end{proof}

Central to the proof of this theorem was Lem. \ref{niceobs}, but it is worth noting that the minimal requirements for proving this lemma were bounds on the probabilities and not necessarily bounds on the elements of the assemblage. We utilised the fact that bounds on the probabilities are obtained from the elements of the assemblage, but if one only bounds the probabilities then our result still follows. We then obtain the following corollary.

\begin{Corollary}\label{corr1}
For the EPR experiment, $f(\epsilon)$-robust correlation-based one-sided self-testing is possible for $f(\epsilon)=24\sqrt{\epsilon}+\epsilon$.
\end{Corollary}

Furthermore, one can also obtain this result using an EPR-steering inequality as we outline in Appendix \ref{app1d} with some minor alterations to the function $f(\epsilon)$. The fact that the function $f(\epsilon)$ in Thm. \ref{thm1} and Cor. \ref{corr1} are the same suggests at the sub-optimality of our analysis, since AST could utilise more information than CST. 


It is now worth commenting on the function $f(\epsilon)$ and contrasting it with results in the standard self-testing literature. In particular, we want to contrast this result with other analytical approaches. This is quite difficult since the measure of closeness to the ideal case is measured in terms of closeness to maximal violation of a Bell inequality and not in terms of elements of an assemblage or individual probabilities. Here we give an indicative comparison between the approach presented here and the current literature. Firstly, McKague, Yang and Scarani developed a means of robust self-testing where if the observed violation of the CHSH inequality is $\epsilon$-close to the maximal violation then the state is $O(\epsilon^{(1/4)})$-close to the ebit \cite{MKYS}. This is a less favourable polynomial than our result which demonstrates $O(\sqrt{\epsilon})$-closeness. On the other hand, the work of Reichardt, Unger and Vazirani \cite{BQC} does demonstrate $O(\sqrt{\epsilon})$-closeness in the state again if $\epsilon$-close to the maximal violation of the CHSH inequality. However, the constant factor in front of the $\sqrt{\epsilon}$ term has been calculated in Ref. \cite{Bancal} to be of the order $10^5$ and our result is several orders of magnitude better even considering the analysis in Appendix \ref{app1d} for a fairer comparison. In various other works \cite{MShi,Cedric,Ivan} more general families of self-testing protocols also demonstrate $O(\sqrt{\epsilon})$-closeness of the physical state to the ebit when the violation is $\epsilon$-far from Tsirelson's bound. We must emphasize that our analysis could definitely be tightened at several stages to lower the constants in $f(\epsilon)$ but EPR-steering already yields an improvement over analytical methods in standard self-testing.

\subsection{Numerical results utilising the SWAP isometry}\label{sec2b}

As demonstrated by the general framework in Refs. \cite{YVBSN} and \cite{Bancal}, numerical methods can be employed to obtain better bounds for self-testing. For reasons that will become clear we will shift focus from AST to CST instead and, in particular, CST based on violation of an EPR-steering inequality. Also, we will not be considering CST in full generality and only seek to establish a bound on the trace distance between the physical and reference states (up to isometries). This will facilitate a direct general comparison with previous works.

We begin by constructing the same SWAP isometry as used in the proof of Thm. \ref{thm1}. As before, it is applied to the physical state $\ket{\psi}$ and again we wish to upper bound the norm in Eq. \ref{tracedist}. Since this is the trace distance between the pure states, $E_{0|0}\ket{\psi}\ket{0_{P'}}+XE_{1|0}\ket{\psi}\ket{1_{P'}}$ and $\ket{\mathcal{A}}\ket{\tilde{\psi}}$, we have that \cite{NC}
\begin{align*}
&\frac{1}{2}\Vert\left(E_{0|0}\ket{\psi}\ket{0_{P'}}+XE_{1|0}\ket{\psi}\ket{1_{P'}}\right)\left(\bra{\psi}E_{0|0}\bra{0_{P'}}+\bra{\psi}E_{1|0}X\bra{1_{P'}}\right)-\ket{\mathcal{A}}\bra{\mathcal{A}}\otimes\ket{\tilde{\psi}}\bra{\tilde{\psi}}\Vert_{1}\leq\sqrt{1-(F^*)^{2}}
\end{align*}
where $F^*=\textrm{max }F$ such that
\begin{align}
F&=\sqrt{\bra{\mathcal{A}}\bra{\tilde{\psi}}\left(E_{0|0}\ket{\psi}\ket{0_{P'}}+XE_{1|0}\ket{\psi}\ket{1_{P'}}\right)\left(\bra{\psi}E_{0|0}\bra{0_{P'}}+\bra{\psi}E_{1|0}X\bra{1_{P'}}\right)\ket{\mathcal{A}}\ket{\tilde{\psi}}}\nonumber\\
&=\frac{1}{\sqrt{2}}\sqrt{\bra{\mathcal{A}}\left(\bra{0_{C}}E_{0|0}\ket{\psi}+\bra{1_{C}}XE_{1|0}\ket{\psi}\right)\left(\bra{\psi}E_{0|0}\ket{0_{C}}+\bra{\psi}E_{1|0}X\ket{1_{C}}\right)\ket{\mathcal{A}}}\nonumber.
\end{align}
Inspired by the work in Refs. \cite{YVBSN} and \cite{Bancal}, instead of bounding the quantity $F$, we wish to bound another quantity $G$ which is the \textit{singlet fidelity}. For $\ket{\tilde{\psi}}=\frac{1}{\sqrt{2}}\left(\ket{0_{C}0_{P'}}+\ket{1_{C}1_{P'}}\right)$, this quantity is defined as
\begin{align}
G&=\bra{\tilde{\psi}}\textrm{tr}_{P}\left[\left(E_{0|0}\ket{\psi}\ket{0_{P'}}+XE_{1|0}\ket{\psi}\ket{1_{P'}}\right)\left(\bra{\psi}E_{0|0}\bra{0_{P'}}+\bra{\psi}E_{1|0}X\bra{1_{P'}}\right)\right]\ket{\tilde{\psi}}\nonumber\\
&=\frac{1}{2}\left(\bra{0_{C}}\sigma_{0|0}\ket{0_{C}}+2\bra{0_{C}}(\sigma_{0|1,0|0}-\sigma_{0|0,0|1,0|0})\ket{1_{C}}+2\bra{1_{C}}(\sigma_{0|0,0|1}-\sigma_{0|0,0|1,0|0})\ket{0_{C}}+\bra{1_{C}}(\rho_{C}-\sigma_{0|0})\ket{1_{C}}\right)\nonumber
\end{align}
such that $\sigma_{0|1,0|0}=\sigma_{0|0,0|1}^{\dagger}=\textrm{tr}_{P}(E_{0|1}E_{0|0}\ket{\psi}\bra{\psi})$ and $\sigma_{0|0,0|1,0|0}=\textrm{tr}_{P}(E_{0|0}E_{0|1}E_{0|0}\ket{\psi}\bra{\psi})$. The above two quantities are related through $(F^*)^{2}\geq 2G-1$ as shown in Ref. \cite{Bancal}. 

The goal is now to give a lower bound to $G$ given constraints on the assemblage. In fact, to facilitate comparison with previous work, we will use the violation of the CHSH inequality to impose these constraints. Every Bell inequality gives an EPR steering inequality when assuming the form of the measurements on the trusted side. If on the client's side we assume the measurements that give the maximal violation of the CHSH inequality for the assemblage generated in the EPR experiment the CHSH expression, denoted by $\textrm{tr} S$, can be written as
\begin{align}
\textrm{tr}S&=\textrm{tr}\frac{1}{\sqrt{2}}\left((\tau_{z}+\tau_{x})(\sigma_{0|0}-\sigma_{1|0})+(\tau_{z}+\tau_{x})(\sigma_{0|1}-\sigma_{1|1})+(\tau_{z}-\tau_{x})(\sigma_{0|0}-\sigma_{1|0})-(\tau_{z}-\tau_{x})(\sigma_{0|1}-\sigma_{1|1})\right) \nonumber\\
&=\textrm{tr}\left(\sqrt{2}\tau_{z}(\sigma_{0|0}-\sigma_{1|0})+\sqrt{2}\tau_{x}(\sigma_{0|1}-\sigma_{1|1})\right) = \textrm{tr}\left(\sqrt{2}\tau_{z}(2\sigma_{0|0} - \rho_C) + \sqrt{2}\tau_{x}(2\sigma_{0|1} - \rho_C)\right) = 2\sqrt{2},\nonumber
\end{align}
where the last bound is Tsirelson's bound. The measurements that the client makes are measurements of the observables in the set $\{{1}/\sqrt{2}(\tau_{z}\pm\tau_{x})\}$. We then have the constraint that $\textrm{tr}S\geq 2\sqrt{2}-\eta$ for a near-maximal violation.

We now want a numerical method of minimising the singlet fidelity $G$ (so as to give a lower bound) such that $\textrm{tr}S\geq 2\sqrt{2}-\eta$. This method is given by the following semi-definite program (SDP):
\begin{align}
\label{eq:sdp}
\textrm{minimize }&\textrm{tr}(M^{T}\Gamma)=G\\ \nonumber
\textrm{subject to: }&\Gamma\geq 0,\\ \nonumber
&\textrm{tr}(N^{T}\Gamma)=\textrm{tr}B\geq 2\sqrt{2}-\eta,
\end{align}
where
\begin{align*}
\Gamma&=\left( \begin{array}{cccc}
\rho_{C} & \sigma_{0|0} & \sigma_{0|1} & \sigma_{0|0,0|1}\\
\sigma_{0|0} & \sigma_{0|0} & \sigma_{0|1,0|0} & \sigma_{0|0,0|1,0|0}\\
\sigma_{0|1} & \sigma_{0|0,0|1} & \sigma_{0|1} & \sigma_{0|0,0|1}\\
\sigma_{0|1,0|0} & \sigma_{0|0,0|1,0|0} & \sigma_{0|1,0|0} & \sigma_{0|0,0|1,0|0}\end{array} \right),&
M&=\frac{1}{2}\left( \begin{array}{cccc}
W & \textbf{0} & \textbf{0} & Y\\
\textbf{0} & \tau_{z} & \textbf{0} & \textbf{0}\\
\textbf{0} & \textbf{0} & \textbf{0} & \textbf{0}\\
Y^{T} & \textbf{0} & \textbf{0} & -2\tau_{x}\end{array} \right),&
N&=2\sqrt{2}\left( \begin{array}{cccc}
\frac{-\tau_x - \tau_z}{2} & \textbf{0} & \textbf{0} & \textbf{0}\\
\textbf{0} & \tau_{z} & \textbf{0} & \textbf{0}\\
\textbf{0} & \textbf{0} & \tau_{x} & \textbf{0}\\
\textbf{0} & \textbf{0} & \textbf{0} & \textbf{0}\end{array} \right),
\end{align*}
such that $W = \left( \begin{smallmatrix} 0&0\\0&1 \end{smallmatrix} \right)$, $Y = \left( \begin{smallmatrix} 0&0\\2&0 \end{smallmatrix} \right)$ and $\textbf{0}$ is a $2$-by-$2$ matrix of all zeroes. We constrain $\Gamma$ in the optimization to be positive semi-definite and not that each sub-matrix of $\Gamma$ corresponding to something like an element of an assemblage is a valid quantum object. It actually turns out that all assemblages that satisfy no-signalling can be realised in quantum theory \cite{Gisin,Hughston}. Discussion of this point is beyond the scope of this paper as all we wish to do is give a lower bound on the value of $G$ therefore just imposing $\Gamma\geq 0$ gives such a bound. 

Before giving an indication of the results of the above SDP, we still need to show that $\Gamma\geq 0$. We do this by showing that $\Gamma$ is a Gramian matrix and all Gramian matrices are positive semi-definite. First observe that entries of $\Gamma$ are of the form $\Gamma_{lm}=\bra{i_{C}}\sigma\ket{j_{C}}$ for $\sigma\in\{\rho_{C},\sigma_{0|0},\sigma_{0|1},\sigma_{0|1,0|0},\sigma_{0|0,0|1},\sigma_{0|0,0|1,0|0}\}$. By cyclicity of the partial trace we can also write $\sigma=\textrm{tr}_{P}(F\ket{\psi}\bra{\psi}G^{\dagger})$ for $F$, $G\in\{\mathbb{I}_{P},E_{0|0},E_{0|1},E_{0|1}E_{0|0}\}$. We now note that
\begin{align*}
\bra{i_{C}}\sigma\ket{j_{C}}&=\sum_{\ket{y}\in\mathcal{H}_{P}}\bra{i_{C}}\bra{y}F\ket{\psi}\bra{\psi}G^{\dagger}\ket{y}\ket{j_{C}}\\
&=\left(\sum_{\ket{y}\in\mathcal{H}_{P}}\bra{i_{C}}\bra{y}F\ket{\psi}\bra{y}\right)\left(\sum_{\ket{y'}\in\mathcal{H}_{P}}\bra{\psi}G^{\dagger}\ket{y'}\ket{j_{C}}\ket{y'}\right)\\
&=\sum_{y}\alpha_{y}\bra{y}\sum_{y'}\alpha_{y'}^{*}\ket{y'}\\
&=\langle{{u}}\ket{v}
\end{align*}
where $\{\ket{y}\}$ is an orthonormal basis in $\mathcal{H}_{P}$ such that $\langle y'\ket{y}=\delta^{y}_{y'}$ and $\alpha_{y}=\bra{i_{C}}\bra{y}F\ket{\psi}$ is some scalar. Since the elements of $\Gamma$ are all the inner product of vectors associated with a row and column, $\Gamma=V^{\dagger}V$ where $V$ has column vectors associated with the vectors $v$. Therefore, $\Gamma$ is Gramian. This then makes the above optimization problem a completely valid problem for lower bounding $G$. We further note that matrix $\Gamma$ represents the EPR-steering analogue of the moment matrix in the Navascu\'{e}s- Pironio-Ac\'{i}n (NPA) hierarchy \cite{NPA} which is useful for approximating the set of quantum correlations \footnote{In principle, we could mimic the NPA hierarchy by constructing matrices $\Gamma$ with elements corresponding to assemblage elements with longer sequences of measurements on the provider's side. However, due to the work in Ref. \cite{Gonzalo}, having the client's system be two-dimensional already essentially puts us in the first level of the hierarchy without the need to go higher.}.

In Fig. \ref{fig:fig3} we plot the lower bound on $G$ achieved through this method and then compare it to the value obtained through the method of Bancal \textit{et al} in Ref. \cite{Bancal}. In both cases the violation of the CHSH inequality is lower-bounded by $2\sqrt{2}-\eta$, and we clearly see that the lower-bound is more favourable for our optimization through EPR-steering as compared to full device-independence. For the case of EPR-steering we observed that the plot can be lower-bounded by the function $1-\eta/\sqrt{2}$ whereas the plot for device-independence is lower-bounded by $1-5\eta/4$. Respectively, these functions give an upper bound on $D(\ket{\Phi}\bra{\Phi},\ket{\mathcal{A}}\bra{\mathcal{A}}\otimes\ket{\tilde{\psi}}\bra{\tilde{\psi}})$ of $2^{\frac{1}{4}}\sqrt{\eta}\leq 1.19\sqrt{\eta}$ and $\sqrt{10}/2\sqrt{\eta}\leq1.59\sqrt{\eta}$. The difference between these two approaches is not as dramatic as the difference in the analytical approaches. However, these results just highlight that the analytical approaches are quite sub-optimal for both EPR-steering and device-independent self-testing. 

\begin{figure}
  \centering
    \includegraphics[width=0.45\textwidth]{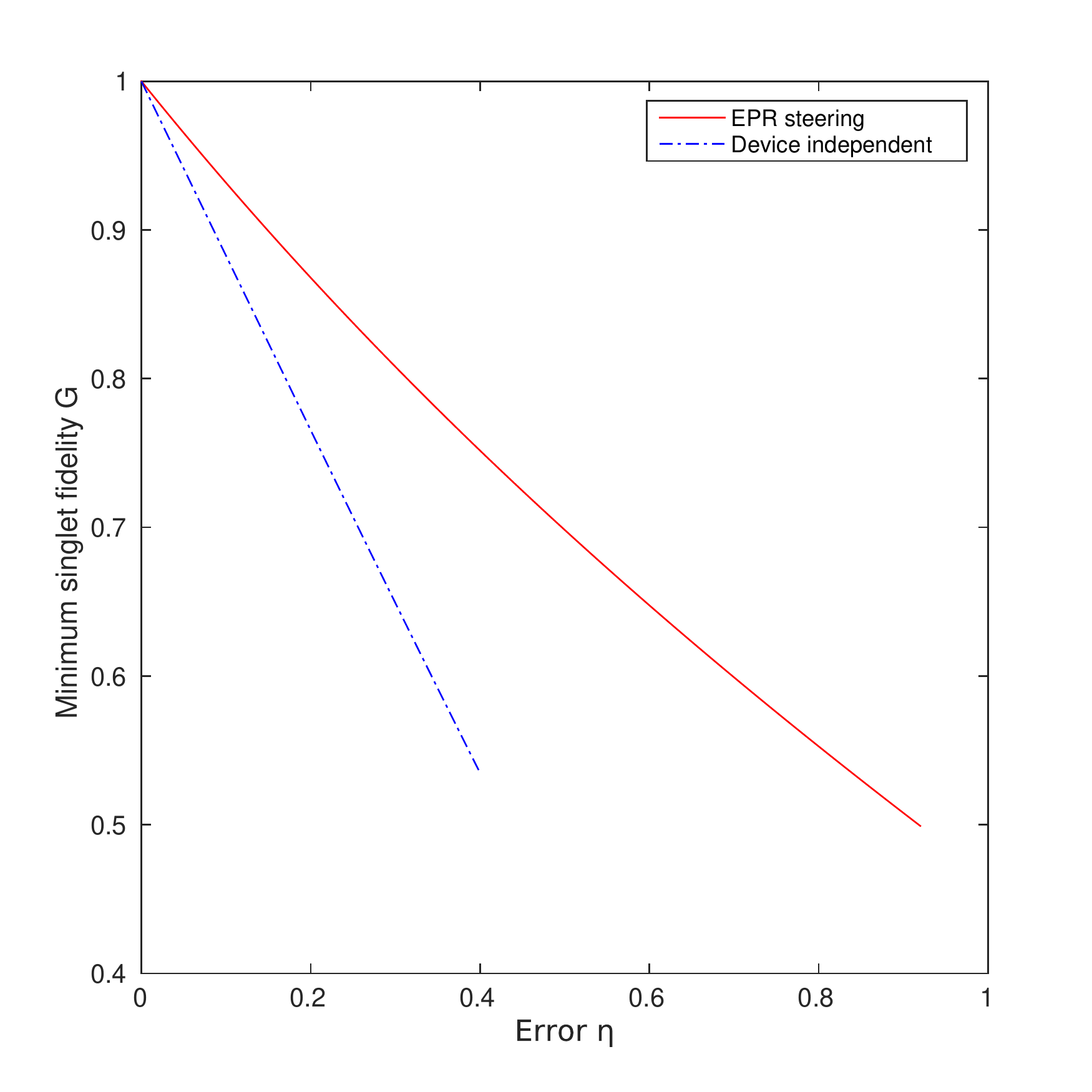}
    \caption{A graph numerically comparing self-testing of the ebit in a device-independent manner to our method based on EPR-steering. The error $\eta$ is the distance from the maximal violation of the CHSH inequality.\label{fig:fig3}}
\end{figure}

\subsection{Optimality of the SWAP isometry}\label{sec2c}

Both the analytical and numerical approaches have utilised the same SWAP isometry. While constructing this isometry demonstrates in a clear and simple manner that self-testing is possible, it is natural to ask if there may be more useful isometries that give a different error scaling for our particular scenario? In particular, can we do better than the $\sqrt{\epsilon}$ in the function $f(\epsilon)$ for $f(\epsilon)$-AST? As we have already shown in Sec. \ref{sec1}, in general this is not possible but the example demonstrating this is somewhat contrived. That is, we are trying to self-test a two-qubit state but assume that the Hilbert space of the client is three-dimensional. We wish to ask if $O(\epsilon)$-AST is possible in the particular example of the EPR experiment? In this section we will show that this is not possible and the best we can hope for is $O(\sqrt{\epsilon})$-AST which we have already established is possible. 

As a side note, in Appendix \ref{app1e} we show that the trace distance between the physical and reference states in the EPR experiment can be $O(\epsilon)$ for some isometries. We emphasize that this trace distance between physical and reference states (condition given in the first line of Eq. \ref{eq:defAST}) only amounts to part of the criteria for AST. The other part of the criteria (the second line of Eq. \ref{eq:defAST}) rules out many isometries that might give the optimal trace distance between physical and reference states only. With this in mind we want to bound the expression in Eq. \ref{condition} for all possible isometries given $\epsilon$-closeness between the elements of the physical and reference assemblages. In particular, we give an example of a physical experiment where $\epsilon$-closeness for the assemblages is satisfied but for all isometries, the smallest value of Eq. \ref{condition} is $O(\sqrt{\epsilon})$. 
\begin{Example}
The physical state is
\begin{equation*}
\ket{\psi}=\frac{1}{\sqrt{2}}\left(\sqrt{1-\epsilon}\ket{0_{C}0_{P}}+\sqrt{\epsilon}\ket{1_{C}1_{P}}\right)\ket{0_{P'}}+\frac{1}{\sqrt{2}}\left(\sqrt{\epsilon}\ket{0_{C}0_{P}}+\sqrt{1-\epsilon}\ket{1_{C}1_{P}}\right)\ket{1_{P'}}
\end{equation*}
where $P$ and $P'$ denote two qubits that the provider has in their possession, thus $\rho_{C}=\frac{1}{2}\mathbb{I}_{C}$. The physical measurements are $E_{0|0}=\mathbb{I}_{P}\otimes\ket{0_{P'}}\bra{0_{P'}}$, $E_{1|0}=\mathbb{I}_{P}\otimes\ket{1_{P'}}\bra{1_{P'}}$, $E_{0|1}=\ket{+_{P}}\bra{+_{P}}\otimes\ket{+_{P'}}\bra{+_{P'}}+\ket{-_{P}}\bra{-_{P}}\otimes\ket{-_{P'}}\bra{-_{P'}}$ and $E_{1|1}=\ket{+_{P}}\bra{+_{P}}\otimes\ket{-_{P'}}\bra{-_{P'}}+\ket{-_{P}}\bra{-_{P}}\otimes\ket{+_{P'}}\bra{+_{P'}}$. These physical measurements on the state produce the following assemblage elements:
\begin{align}
\sigma_{0|0}&=\frac{(1-\epsilon)}{2}\ket{0_{C}}\bra{0_{C}}+\frac{\epsilon}{2}\ket{1_{C}}\bra{1_{C}},&
\sigma_{1|0}&=\frac{(1-\epsilon)}{2}\ket{1_{C}}\bra{1_{C}}+\frac{\epsilon}{2}\ket{0_{C}}\bra{0_{C}},\nonumber\\
\sigma_{0|1}&=\frac{1}{2}\ket{+_{C}}\bra{+_{C}},&
\sigma_{1|1}&=\frac{1}{2}\ket{-_{C}}\bra{-_{C}}.\nonumber
\end{align}
We see then that $D(\rho_{C},\tilde{\rho}_{C})=0$ and $\Vert\sigma_{a|x}-\tilde{\sigma}_{a|x}\Vert\leq\epsilon$ for all $a$, $x$.

We now show that $\Vert\ket{\Phi,E_{0|0}}\bra{\Phi,E_{0|0}}-\ket{\mathcal{A}}\bra{\mathcal{A}}\otimes\tilde{E}_{0|0}\ket{\tilde{\psi}}\bra{\tilde{\psi}} \tilde{E}_{0|0}\Vert_{1}\geq\sqrt{\epsilon}$ for all possible isometries $\Phi$. By considering all possible isometries we have
\begin{equation*}
\ket{\Phi,E_{0|0}}=UE_{0|0}\ket{\psi}\ket{\hat{0}}=\frac{1}{\sqrt{2}}U\left(\sqrt{1-\epsilon}\ket{0_{C}0_{P}}+\sqrt{\epsilon}\ket{1_{C}1_{P}}\right)\ket{0_{P'}}\ket{\hat{0}}=\frac{1}{\sqrt{2}}\ket{\epsilon},
\end{equation*}
for $\ket{\epsilon}=U\left(\sqrt{1-\epsilon}\ket{0_{C}0_{P}}+\sqrt{\epsilon}\ket{1_{C}1_{P}}\right)\ket{0_{P'}}\ket{\hat{0}}$ and $U$ being a unitary applied jointly to the provider's qubits and the ancillae $\ket{\hat{0}}$. This then allows us to observe that
\begin{equation*}
\Vert\ket{\Phi,E_{0|0}}\bra{\Phi,E_{0|0}}-\ket{\mathcal{A}}\bra{\mathcal{A}}\otimes\tilde{E}_{0|0}\ket{\tilde{\psi}}\bra{\tilde{\psi}}\tilde{E}_{0|0}\Vert_{1}=D(\ket{\epsilon}\bra{\epsilon},\ket{\mathcal{A}}\bra{\mathcal{A}}\otimes\ket{00}\bra{00})=\sqrt{1-\vert\langle{\epsilon}\ket{\mathcal{A}}\ket{00}\vert^{2}}.
\end{equation*}
We see that $\vert\langle{\epsilon}\ket{\mathcal{A}}\ket{00}\vert^{2}=(1-\epsilon)\vert\bra{\mathcal{A}}\bra{0}U\ket{0_{P}}\ket{\hat{0}}\vert^{2}$ which achieves the maximal value of $(1-\epsilon)$. Therefore $\Vert\ket{\Phi,E_{0|0}}\bra{\Phi,E_{0|0}}-\ket{\mathcal{A}}\bra{\mathcal{A}}\otimes\tilde{E}_{0|0}\ket{\tilde{\psi}}\bra{\tilde{\psi}}\tilde{E}_{0|0}\Vert_{1}\geq\sqrt{\epsilon}$ for all possible isometries $\Phi$. 
\end{Example}
This example demonstrates that $O(\epsilon)$-AST is impossible for the EPR experiment and our analytical results are essentially optimal (up to constants).

\section{Self-testing multi-partite states}\label{sec3}

So far all the work presented thus far has been presented within a bipartite format both in terms of the client-provider scenario but also the reference state's Hilbert space being the tensor product of two Hilbert spaces. Due to their utility in various tasks, the self-testing of multi-partite quantum states is also desirable. Within the device-independent self-testing literature there have already been many developments along this line of research (see, e.g. Refs. \cite{McKague, Wu}). In this section we give a brief indication of how to generalise our set-up to the consideration of such states. In Sec. \ref{sec3a} we will discuss the self-testing of tri-partite states and give initial numerical results demonstrating the richness of this scenario. We will briefly sketch in Sec. \ref{sec3b} how EPR-steering could prove useful in establishing a tensor product structure within the provider's Hilbert space. 

\subsection{Self-testing the GHZ state}\label{sec3a}

Already for three parties, how to modify the client-provider set-up opens up new and interesting possibilities. For example, the simplest modification is to have the new, third party be a trusted part of the client's laboratory; the total Hilbert space of the client $\mathcal{H}_{C}$ is now the tensor product of the two Hilbert spaces associated with these two parties. The next possible modification, as shown in Fig. \ref{fig:fig4}, is to have a second untrusted party that after receiving their share of the physical state does not communicate with the initial provider: they only communicate with the client. This restriction establishes a tensor product structure between the two untrusted parties which is useful. 

\begin{figure}
  \centering
    \includegraphics[width=0.33\textwidth]{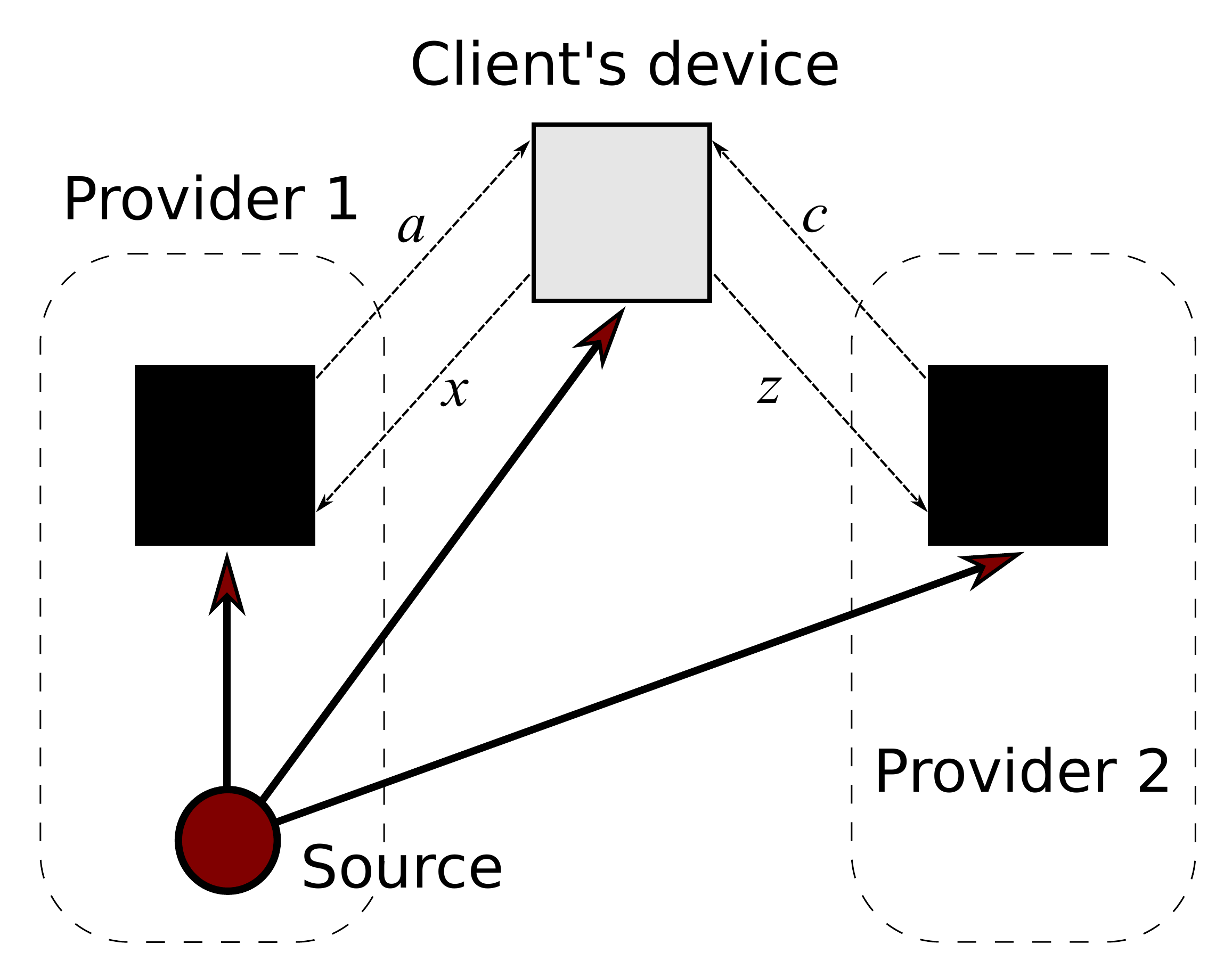}
    \caption{Here we depict the tripartite set-up with three parties where only one is the client, called the $1$-trusted setting in the text. There are two non-communicating providers and we assume without loss of generality that one of them generates a quantum state and sends one part to the client and another to the other provider. The client may communicate with each provider individually and ask them to perform measurements.\label{fig:fig4}}
\end{figure}

To illustrate the interesting differences between the bipartite and tri-partite cases, we look at the example of self-testing the Greenberger-Horne-Zeilinger (GHZ) state $\ket{\tilde{\psi}}=1/\sqrt{2}\left(\ket{\Psi}\ket{+_{3}}+\ket{\Psi'}\ket{-_{3}}\right)$ where $\ket{\Psi}=1/\sqrt{2}(\ket{0_{1}0_{2}}-\ket{1_{1}1_{2}})$ and $\ket{\Psi'}=1/{\sqrt{2}}(\ket{0_{1}1_{2}}+\ket{1_{1}0_{2}})$ with subscripts denoting the number of the qubit. In the scenario with two trusted parties (that together form the client), a qubit is sent from the provider to each of these parties (say, qubits $1$ and $2$ are sent); we will call this scenario the $2$\textit{-trusted setting}. In the other scenario with two non-communicating untrusted providers, a qubit (say, qubit $1$) is sent to the client; we will call this scenario the $1$\textit{-trusted setting}. These different scenarios correspond to different types of multipartite EPR-steering introduced in Ref. \cite{DaniPaul}.

We now describe the reference experiments for both settings for the state $\ket{\tilde{\psi}}$. In the case of the $2$-trusted setting, as in the EPR experiment, the provider claims to make measurements $\tilde{E}_{j|0}=\ket{j}\bra{j}$ for $j\in\{0,1\}$ as well as $\tilde{E}_{0|1}=\ket{+}\bra{+}$ and $\tilde{E}_{1|1}=\ket{-}\bra{-}$. The assemblage for the two trusted parties has elements
\begin{align}
\tilde{\sigma}_{0|0}&=\frac{1}{4}(\ket{\Psi}+\ket{\Psi'})(\bra{\Psi}+\bra{\Psi'}),&
\tilde{\sigma}_{1|0}&=\frac{1}{4}(\ket{\Psi}-\ket{\Psi'})(\bra{\Psi}-\bra{\Psi'}),\nonumber\\
\tilde{\sigma}_{0|1}&=\frac{1}{2}\ket{\Psi}\bra{\Psi},&
\tilde{\sigma}_{1|1}&=\frac{1}{2}\ket{\Psi'}\bra{\Psi'}.\nonumber
\end{align}
For the $1$-trusted setting, in addition to the provider claiming to making the above measurements, the second untrusted party, or second provider claims also to make the same measurements, which we denote by $\tilde{E}_{c|z}$ for $c$, $z\in\{0,1\}$. The assemblage will be $\{\tilde{\sigma}_{a,c|x,z}\}_{a,c,x,z}$ where each element is $\tilde{\sigma}_{a,c|x,z}=\textrm{tr}_{P}(\mathbb{I}_{C}\otimes \tilde{E}_{c|z}\otimes \tilde{E}_{a|x}\ket{\tilde{\psi}}\bra{\tilde{\psi}})$. The assemblage for the one trusted party will have $16$ elements but for the sake of brevity we will not write out the elements.

We then wish to self-test this reference experiment when the elements of the physical assemblage are close to the elements of the ideal, reference experiment. Instead of doing this, we will mimic the numerical approach in Sec. \ref{sec2b} by considering the GHZ-Mermin inequality \cite{GHZMermin} adapted to the $1$-trusted and $2$-trusted scenarios. Utilising the notation of $\tau_{x}$ and $\tau_{z}$ for the Pauli-$X$ and Pauli-$Z$ matrices respectively, for the $2$-trusted and $1$-trusted settings, the inequalities respectively are:
\begin{align*}
\textrm{tr}B_{2}=&2\textrm{tr}\left((\tau_{z}\otimes\tau_{z})(2\sigma_{0|1}-\rho_C)+(\tau_{x}\otimes\tau_{z})(2\sigma_{0|0}-\rho_c)+(\tau_{z}\otimes\tau_{x})(2\sigma_{0|0}-\rho_C)-(\tau_{x}\otimes\tau_{x})(2\sigma_{0|1}-\rho_C)\right)\leq 2,\\
\textrm{tr}B_{1}=&2\textrm{tr}\left(\tau_{z}(\sigma_{00|01}-\sigma_{01|01}-\sigma_{10|01}+\sigma_{11|01})+\tau_{x}(\sigma_{00|00}+\sigma_{11|00}-\sigma_{01|00}-\sigma_{10|00})\right) + \\ & 2\textrm{tr}\left(\tau_{z}(\sigma_{00|10}+\sigma_{11|10}-\sigma_{01|10}-\sigma_{10|10})-\tau_{x}(\sigma_{00|11}+\sigma_{11|11}-\sigma_{01|11}-\sigma_{10|11})\right)\leq 2.
\end{align*}
The maximal quantum violation of these inequalities is $4$. We now aim to carry out self-testing if the physical experiment achieves a violation of $4-\eta$. For the untrusted parties, we implement the SWAP isometry to each of their systems as outlined in Sec. \ref{sec2a}. For the $2$-trusted setting, the physical state $\ket{\psi}$ gets mapped to $\ket{\psi'}=E_{0|0}\ket{\psi}\ket{0_P'}+XE_{1|0}\ket{\psi}\ket{1_P'}$. In the $1$-trusted setting, the physical state $\ket{\psi}$ gets mapped to
\begin{equation*}
\ket{\psi''}=E_{0|0}F_{0|0}\ket{\psi}\ket{0_{P'}}\ket{0_{P''}}+XE_{1|0}F_{0|0}\ket{\psi}\ket{1_{P'}}\ket{0_{P''}}+E_{0|0}X'F_{1|0}\ket{\psi}\ket{0_{P'}}\ket{1_{P''}}+XE_{1|0}X'F_{1|0}\ket{\psi}\ket{1_{P'}}\ket{1_{P''}}
\end{equation*}
where $F_{c|z}$ is the physical measurement made by the second untrusted party, $X'=2F_{0|1}-\mathbb{I}$ and $P'$ denotes the ancilla qubit introduced for one party and $P''$ for the other party.

Our figure of merit for closeness between the physical and reference states is the \textit{GHZ fidelity} which for the $2$-trusted and $1$-trusted settings is $G_{2}$ and $G_{1}$ respectively where
\begin{align*}
G_{2}&=\bra{\tilde{\psi}}\textrm{tr}_{P}\left(\ket{\psi'}\bra{\psi'}\right)\ket{\tilde{\psi}},\nonumber\\
G_{1}&=\bra{\tilde{\psi}}\textrm{tr}_{P}\left(\ket{\psi''}\bra{\psi''}\right)\ket{\tilde{\psi}},
\end{align*}
where in both cases we trace out the provider's (providers') Hilbert space(s) $\mathcal{H}_{P}$. Now we minimize $G_{2}$ while $\textrm{tr}B_{2}\geq 4-\eta$ and minimize $G_{1}$ such that $\textrm{tr}B_{2}\geq 4-\eta$. These problems again can be lower-bounded by an SDP and in Fig. \ref{fig:fig5} we give numerical values obtained with these minimization problems. This case is numerically more expensive than the simple self-testing of the EPR experiment and for tackling it we used the SDP procedures described in Ref. \cite{Wittek}. We also compare our results to those obtained in the device-independent setting where all three parties are not trusted but the violation of the GHZ-Mermin inequality is $4-\eta$. We see that the GHZ fidelity increases when we trust more parties. Interestingly, we can see that the curve for $1$-trusted scenario is obviously closer to the curve of $2$-trusted scenario than to the device-independent one. This may hint that multi-partite EPR-steering behaves quite differently to quantum non-locality. However, to draw this conclusion from self-testing one would have to pursue more rigorous research, since we have only obtained numerical lower bounds on the GHZ fidelity using only one specific isometry.

\begin{figure}
  \centering
    \includegraphics[width=0.45\textwidth]{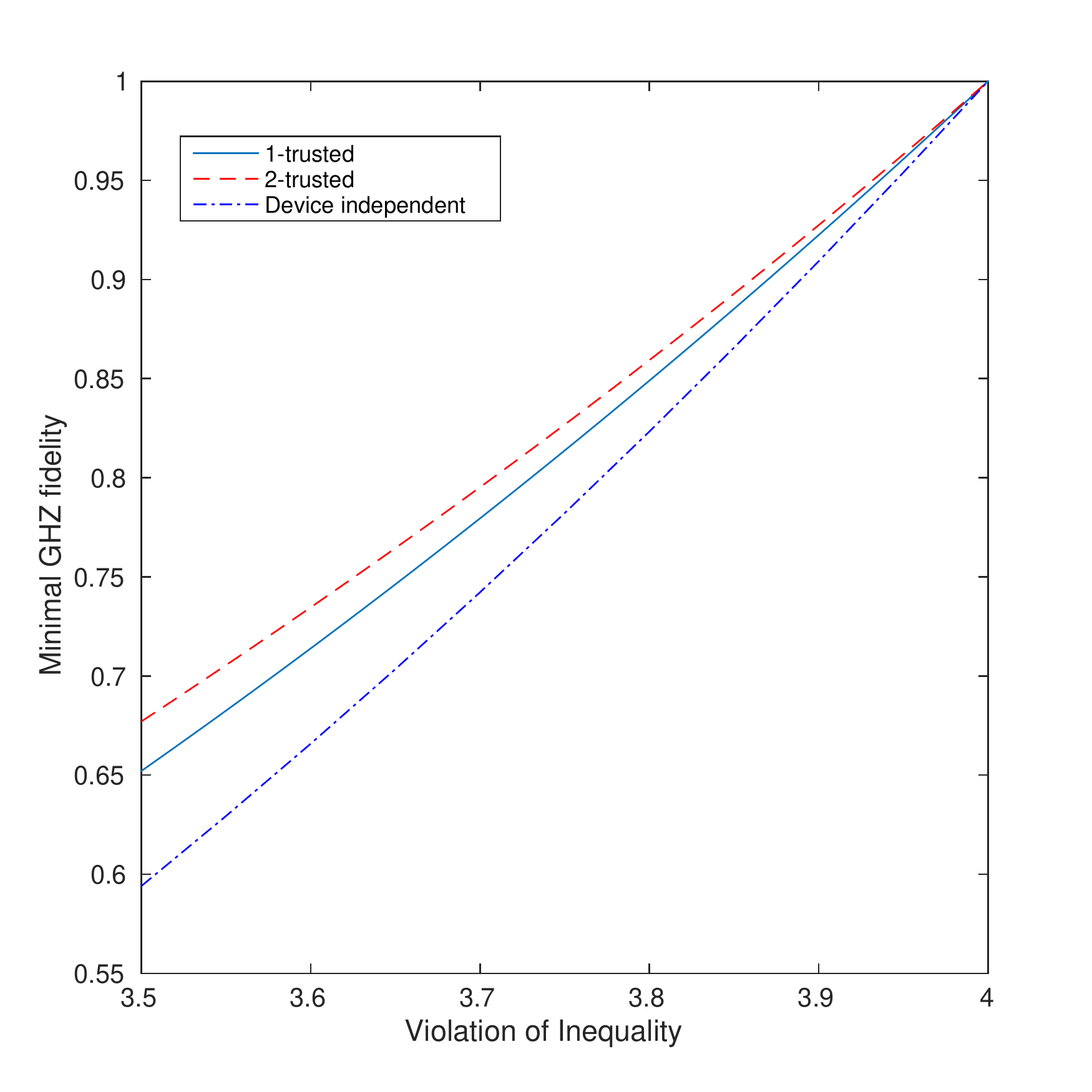}
    \caption{A graph numerically comparing the minimum GHZ fidelity for a given violation of the GHZ-Mermin inequality for different levels of trust in the devices. We observe that the line for the $1$-trusted setting is closer to the $2$-trusted setting than device-independence. In future work we will aim to understand if there is fundamental reason for this.\label{fig:fig5}}
\end{figure}

\subsection{Establishing a tensor product structure}\label{sec3b}

The previous section hints at what might be the most useful aspect of self-testing through EPR-steering: establishing a tensor product structure in the provider's Hilbert space. In the work of Reichardt, Unger and Vazirani, a method is presented for self-testing many copies of the ebit between two untrusted parties \cite{BQC}. This testing is achieved through measurements made in sequence. Recent work has established the same feat but now with measurements being made at the same time, thus giving a more general result \cite{Parallel}. The difficulty in establishing that the two untrusted parties have multiple copies of the ebit is to establish that (up to isometries) the Hilbert spaces of the parties decompose as a tensor product of several $2$-dimensional Hilbert spaces: in each sub-space there is one-half of an ebit.

We now remark that EPR-steering offers a useful simplification in achieving the same task of identifying a tensor product structure. Note that in the trusted laboratory a tensor product structure is known: the client knows they have, say, two qubits. If the assemblage for each qubit is close to the ideal case of being one half of an ebit, then we may use Lem. \ref{niceobs} to ``transfer" the physical operations on the untrusted side to one of the qubits on the trusted side. We also note that this observation forms part of the basis of the work presented in Ref. \cite{Edinburgh}, in the context of verification of quantum computation.

To be more exact, we now have the client's Hilbert space being constructed from a tensor product of $N$ two-dimensional Hilbert spaces, i.e. $\mathcal{H}_{C}=\bigotimes_{i=1}^{N}\mathcal{H}_{C_{i}}$ where $\mathcal{H}_{C_{i}}=\mathbb{C}^{2}$. We now have a modified form of the EPR experiment with the reference state being $\vert\tilde{\psi}\rangle=\bigotimes_{i=1}^{N}\vert\tilde{\psi}_{i}\rangle\in\bigotimes_{i=1}^{N}\mathcal{H}_{P_{i}}\otimes\mathcal{H}_{C_{i}}$ for each $\vert\tilde{\psi}_{i}\rangle=\frac{1}{\sqrt{2}}(\vert 00\rangle+\vert 11\rangle)\in\mathcal{H}_{P_{i}}\otimes\mathcal{H}_{C_{i}}$. That is, in the reference experiment, the provider's Hilbert space has a tensor product structure. For each Hilbert space $\mathcal{H}_{P_{i}}$, there is a projective measurement with projectors $\tilde{E}_{a_{i}|x_{i}}$ acting on that space where $a_{i}$, $x_{i}\in\{0,1\}$ and these projectors are the qubit projectors in the EPR experiment. Therefore, the total reference projector is of the form $\bigotimes_{i=1}^{N}\tilde{E}_{a_{i}|x_{i}}$ which act on the Hilbert space $\bigotimes_{i=1}^{N}\mathcal{H}_{P_{i}}$. In this case, the measurement choices and outcomes are bit-strings $\textbf{x}:=(x_{1},x_{2},...,x_{N})$ and $\textbf{a}:=(a_{1},a_{2},...,a_{N})$ respectively. We call this reference experiment the \textit{N-pair EPR experiment} and we are now in a position to generalise Lem. \ref{niceobs}.

\begin{lemma}\label{Nniceobs}
For the N-pair EPR experiment, if for all $i$, $\Vert\sigma_{\textbf{a}|\textbf{x}}-\tilde{\sigma}_{\textbf{a}|\textbf{x}}\Vert_{1}\leq\epsilon$ and $D(\rho_{C},\tilde{\rho}_{C})\leq\epsilon$ where $\tilde{\sigma}_{\textbf{a}|\textbf{x}}=\bigotimes_{i=1}^{N}\tilde{\sigma}_{a_{i}|x_{i}}$ and $\tilde{\rho}_{C}=\bigotimes_{i=1}^{N}\frac{\mathbb{I_{C}}}{2}$ then
\begin{equation}
\Vert\mathbb{I}_{C}\otimes E_{\textbf{a}|\textbf{x}}\ket{\psi}-\bigotimes_{i=1}^{N}\tilde{E}_{a_{i}|x_{i}}\otimes\mathbb{I}_{P}\ket{\psi}\Vert\leq 2\sqrt{\epsilon}.
\end{equation}
\end{lemma}

The proof of this lemma is almost identical to the proof of Lem. \ref{niceobs} and so we will leave it out from our discussion. A nice relaxation of the conditions of the above lemma is to insist that each observed element of an assemblage $\sigma_{a_{i}|x_{i}}$ is $\epsilon$-close to $\tilde{\sigma}_{a_{i}|x_{i}}$ and still recover a similar result. This requires a little bit more work since we have not been specific in how we model the provider's measurements. For example, we have not stipulated whether the probability distribution $p(\textbf{a}|\textbf{x})=\textrm{tr}(\sigma_{\textbf{a}|\textbf{x}})$ satisfies the no-signalling principle. Furthermore, even if these probabilities satisfy this principle, it does not immediately enforce a constraint on the behaviour of the measurements. For the sake of brevity we will not address this issue in this work. It remains to point out that Lem. \ref{Nniceobs} can be used to develop a result for self-testing (cf Ref. \cite{Edinburgh}).

\section{Discussion}\label{sec4}

In our work we have explored the possibilities of self-testing quantum states and measurements based on bipartite (and multi-partite) EPR-steering. We have shown that the framework allows for a broad range of tools for performing self-testing. One can use state tomography on part of the state and use this information to get more useful analytical methods. Or, indeed, one only needs to use the probabilities of outcomes for certain fixed (and known) measurements. Furthermore, self-testing can be based solely on the near-maximal violation of an EPR-steering inequality. We compared these approaches to the standard device-independent approach and demonstrated that EPR-steering simplifies proofs and gives more useful bounds for robustness. We hope that this could be used in future experiments where states produced are quite far from ideal but potentially useful for quantum information tasks. However, we note that EPR-steering-based self-testing only really improves the constants in the error terms (for robustness) and not the polynomial of the error, i.e. we can only demonstrate $O(\sqrt{\epsilon})$-AST for the EPR experiment. This highlights that from the point-of-view of self-testing, EPR-steering resembles quantum non-locality and not entanglement verification in which all parties are trusted.

In future work, we wish to explore the self-testing of other quantum states. For example, we can show that similar techniques as outlined in this work can be used to self-test partially entangled two-qubit states. We would like to give a general framework in which many examples of states and measurements can be self-tested. This would be something akin to the work of Yang \textit{et al} \cite{YVBSN} that utilizes the NPA hierarchy of SDPs. Recent work by Kogias \textit{et al} \cite{Giannis} could prove useful in this aim. In addition to this, our work has hinted at the interesting possibilities for studying self-testing based on EPR-steering in the multipartite case. In future work we will investigate adapting our techniques to general multipartite states. For example, the general multipartite GHZ state can be self-testing by adapting the family of Bell inequalities found in Refs. \cite{GenGHZ1,GenGHZ2,GenGHZ3}.

Also, it would be interesting to try to establish some new insights in the fundamental relations between non-locality and EPR-steering using self-testing. It is possible that self-testing could be a useful tool for exploring their similarities and differences, especially given interesting new developments for multi-partite EPR steering \cite{postquantum}.

One may question our use of the Schatten $1$-norm as a measure of distance between elements of a reference and physical assemblage. For example, the Schatten $2$-norm is a lower bound on the $1$-norm so could be a more useful measure of closeness. It may be worthwhile to explore this possibility but we note that the argument for the impossibility of $O(\epsilon)$-AST for the EPR experiment in Sec. \ref{sec2c} still applies even if we replace all the distance measures with the $2$-norm.

Finally, it would be interesting to consider relaxing the assumption of systems being independent and identically distributed (i.i.d) and tomography being performed in the asymptotic limit. This would take into account the provider having devices with memory as well as only being given a finite number of systems. In the case of CST, we may use statistical methods to bound the probability that the provider can deviate from their claims and trick us in accepting their claims. For the case of AST, tools from non-i.i.d. quantum information theory might be required which makes the future study of AST interesting from the point-of-view of quantum information.
\newline

\textit{Acknowledgements} - The authors acknowledge useful discussions with Antonio Ac\'{i}n, Paul Skrzypczyk, Daniel Cavalcanti and Peter Wittek. MJH also thanks Nathan Walk for discussions and Petros Wallden, Andru Gheorghiu and Elham Kashefi for discussing their recent independent work in Ref. \cite{Edinburgh} about self-testing based on EPR-steering as applied to the verification of quantum computation. MJH acknowledges support from the EPSRC (through the NQIT Quantum Hub) and the FQXi Large Grants \textit{Thermodynamic vs information theoretic entropies in probabilistic theories} and \textit{Quantum Bayesian networks: the physics of nonlocal events}. IS asknowledges funding from the ERC CoG project QITBOX, the MINECO project FOQUS, the Generalitat de Catalunya (SGR875) and the Ministry of Science of Montenegro (Physics of Nanostructures, Contract No 01-682).

\begin{appendix}
\section{Complex conjugation and assemblages}\label{app1}

In this section we give an example of an assemblage that is altered upon taking the complex conjugation of the state and measurements on the provider's side. The state is $\ket{\psi}=\frac{1}{\sqrt{2}}\left(\ket{0_{C}0_{P}}+i\ket{1_{C}1_{P}}\right)$ and we consider the element of the assemblage generated by the projector $\ket{+_{P}}\bra{+_{P}}$. The element of the assemblage is then $\textrm{tr}_{P}\left(\mathbb{I}_{C}\otimes\ket{+_{P}}\bra{+_{P}}\ket{\psi}\bra{\psi}\right)=\frac{1}{2}\vert +y_{C}\rangle\langle +y_{C}\vert$ for $\ket{+y_{C}}=\frac{1}{\sqrt{2}}\left(\ket{0_{P}}+i\ket{1_{P}}\right)$. We immediately see that upon taking the complex conjugate of the state $\ket{\psi^{*}}$ and projector, the respective element of the assemblage becomes $\textrm{tr}_{P}\left(\mathbb{I}_{C}\otimes\ket{+_{P}}\bra{+_{P}}\ket{\psi^{*}}\bra{\psi^{*}}\right)=\frac{1}{2}\vert -y_{C}\rangle\langle -y_{C}\vert$ $\ket{-y_{C}}=\frac{1}{\sqrt{2}}\left(\ket{0_{P}}-i\ket{1_{P}}\right)$. Therefore if the client measures the element of the assemblage in the basis $\{\ket{\pm y_{C}}\}$, they can differentiate between the two cases of the physical state being $\ket{\psi}$ and its complex conjugate $\ket{\psi^{*}}$.

\section{Obtaining the bound in Eq. \ref{final}}\label{app1c}

We now aim to put a bound on
\begin{equation}\label{newcondition}
\Vert\ket{\Phi,Q_{a|x}}\bra{\Phi,Q_{a|x}}-\ket{\mathcal{A}}\bra{\mathcal{A}}\otimes\tilde{Q}_{a|x}\ket{\tilde{\psi}}\bra{\tilde{\psi}}\otimes \tilde{Q}_{a|x}\Vert_{1}
\end{equation}
where 
\begin{equation}
\ket{\Phi,Q_{a|x}}=E_{0|0}Q_{a|x}\ket{\psi}\ket{0_{P'}}+XE_{1|0}Q_{a|x}\ket{\psi}\ket{1_{P'}}.
\end{equation}
Then we aim to prove the bound in Eq. \ref{final} by expanding out Eq. \ref{newcondition} where $Q_{a|x}\in\{E_{0|1},E_{1|1}\}$ and $\tilde{Q}_{a|x}\in\{\ket{+_{C}}\bra{+_{C}},\ket{-_{C}}\bra{-_{C}}\}$. We focus on the case where $Q_{a|x}=E_{0|1}$ and $\tilde{Q}_{a|x}=\ket{+_{C}}\bra{+_{C}}$ since the other case is essentially yields essentially the same bound for Eq. \ref{newcondition}. We, therefore, wish to find an upper bound for
\begin{align}
&\Vert\left(E_{0|0}E_{0|1}\ket{\psi}\ket{0_{P'}}+XE_{1|0}E_{0|1}\ket{\psi}\ket{1_{P'}}\right)\left(\bra{\psi}E_{0|1}E_{0|0}\bra{0_{P'}}+\bra{\psi}E_{0|1}E_{1|0}X\bra{1_{P'}}\right)-\frac{1}{2}\ket{\mathcal{A}}\bra{\mathcal{A}}\otimes\ket{+_{C}+_{P'}}\bra{+_{C}+_{P'}}\Vert_{1}.\nonumber
\end{align}
Through repeated uses of Lem. \ref{niceobs} we obtain 
\begin{align}
&\Vert\left(E_{0|0}E_{0|1}\ket{\psi}\ket{0_{P'}}+XE_{1|0}E_{0|1}\ket{\psi}\ket{1_{P'}}\right)\left(\bra{\psi}E_{0|1}E_{0|0}\bra{0_{P'}}+\bra{\psi}E_{0|1}E_{1|0}X\bra{1_{P'}}\right)-\frac{1}{2}\ket{\mathcal{A}}\bra{\mathcal{A}}\otimes\ket{+_{C}+_{P'}}\bra{+_{C}+_{P'}}\Vert_{1}\nonumber\\
&\leq 24\sqrt{\epsilon}+\Vert\left(\tilde{E}_{0|1}\tilde{E}_{0|0}\ket{\psi}\ket{0_{P'}}+\tilde{E}_{0|1}\tilde{E}_{1|0}\tau_{X}\ket{\psi}\ket{1_{P'}}\right)\left(\bra{\psi}\tilde{E}_{0|0}\tilde{E}_{0|1}\bra{0_{P'}}+\bra{\psi}\tau_{x}\tilde{E}_{1|0}\tilde{E}_{0|1}\bra{1_{P'}}\right)\nonumber\\
&-\frac{1}{2}\ket{\mathcal{A}}\bra{\mathcal{A}}\otimes\ket{+_{C}+_{P'}}\bra{+_{C}+_{P'}}\Vert_{1}\nonumber\\
&=24\sqrt{\epsilon}+\frac{1}{2}\Vert\left(\ket{+_{C}}\langle 0_{C}\ket{\psi}\ket{0_{P'}}+\ket{+_{C}}\bra{0_{C}}\psi\rangle\ket{1_{P'}}\right)\left(\bra{\psi}0_{C}\rangle\bra{+_{C}}\bra{0_{P'}}+\bra{\psi}0_{C}\rangle\bra{+_{C}}\bra{1_{P'}}\right)-\ket{\mathcal{A}}\bra{\mathcal{A}}\otimes\ket{+_{C}+_{P'}}\bra{+_{C}+_{P'}}\Vert_{1}\nonumber\\
&\leq 24\sqrt{\epsilon}+\frac{1}{2}\Vert 2\langle{0_{C}}\ket{\psi}\bra{\psi}0_{C}\rangle-\ket{\mathcal{A}}\bra{\mathcal{A}}\Vert_{1}.\nonumber
\end{align}
The first inequality is obtained in conjunction with the fact that $\Vert E_{0|0}E_{0|1}\ket{\psi}\ket{0_{P'}}+XE_{1|0}E_{0|1}\ket{\psi}\ket{1_{P'}}\Vert=\sqrt{\bra{\psi}E_{0|1}\ket{\psi}}\leq 1$ and $\Vert \tilde{E}_{0|1}\tilde{E}_{0|0}\ket{\psi}\ket{0_{P'}}+\tilde{E}_{0|1}\tilde{E}_{1|0}\tau_{X}\ket{\psi}\ket{1_{P'}}\Vert=\sqrt{\bra{\psi}0\rangle\bra{0}\psi\rangle}\leq 1$. In the proof of Thm. \ref{thm1} it was shown that $\Vert 2\langle{0_{C}}\ket{\psi}\bra{\psi}0_{C}\rangle-\ket{\mathcal{A}}\bra{\mathcal{A}}\Vert_{1}\leq 2\epsilon$ which then gives us the function $f(\epsilon)$ in Thm. \ref{thm1}.
\section{Robust self-testing based on an EPR-steering inequality}\label{app1d}

In this section we use an EPR-steering inequality to give us a result for CST. In particular, we prove a version of Lem. \ref{niceobs}. Given this, all the steps in Thm. \ref{thm1} apply. The EPR-steering inequality we use is the following
\begin{equation}
\sum_{a|x}\textrm{tr}\left(F_{a|x}\sigma_{a|x}\right)=\textrm{tr}\left({\sqrt{2}}(\sigma_{0|0}+\sigma_{1|0})-(2\tau_{z}-1)(\sigma_{0|0}-\sigma_{1|0})-(2\tau_{x}-1)(\sigma_{0|1}-\sigma_{1|1})\right)\geq 0. 
\end{equation}
This can be written in the simplified form of 
\begin{equation}
\langle\psi\vert \tau_{z}\otimes Z\vert\psi\rangle+\langle\psi\vert \tau_{x}\otimes X\vert\psi\rangle\leq\sqrt{2}
\end{equation}
where $Z=2E_{0|0}-\mathbb{I}$, $X=2E_{0|1}-\mathbb{I}$ with $\tau_{x}$ and $\tau_{z}$ being the Pauli-$X$ and Pauli-$Z$ matrices respectively. It can be readily verified that the EPR experiment violates this inequality and achieves a value of $2$ for the left-hand-side; this is the maximal attainable value. Given near-maximal violation we wish to prove a version of Lem. \ref{niceobs}.
\begin{lemma}\label{niceobsnew}
If $\langle\psi\vert \tau_{z}\otimes Z\vert\psi\rangle+\langle\psi\vert \tau_{x}\otimes X\vert\psi\rangle\geq 2-\eta$ for $1\geq\eta\geq 0$, then
\begin{equation}
\Vert\mathbb{I}_{C}\otimes E_{a|x}\ket{\psi}-\tilde{E}_{a|x}\otimes\mathbb{I}_{P}\ket{\psi}\Vert\leq\sqrt{\eta}
\end{equation}
\end{lemma}

\begin{proof}
From the near-maximal violation of the EPR-steering inequality we have that $\langle\psi\vert \tau_{z}\otimes Z\vert\psi\rangle\geq 1-\eta$ and $\langle\psi\vert \tau_{x}\otimes X\vert\psi\rangle\geq 1-\eta$. We will address the case where $a=x=0$ as all other cases follow the same proof strategy. We first note that we can write $\langle\psi\vert \tau_{z}\otimes Z\vert\psi\rangle$ as $\bra{\psi}(2\tilde{E}_{0|0}-\mathbb{I})(2E_{0|0}-\mathbb{I})\ket{\psi}\geq 1-\eta$. Utilising this, we make a series of simple observations: 
\begin{align}
\Vert\mathbb{I}_{C}\otimes E_{0|0}\ket{\psi}-\tilde{E}_{0|0}\otimes\mathbb{I}_{P}\ket{\psi}\Vert&=\sqrt{\bra{\psi}\mathbb{I}_{C}\otimes E_{0|0}\ket{\psi}+\bra{\psi} \tilde{E}_{0|0}\otimes\mathbb{I}_{P}\ket{\psi}-2\bra{\psi}\tilde{E}_{0|0}\otimes E_{0|0}\ket{\psi}}\nonumber\\
&=\sqrt{\frac{1}{2}-\frac{1}{2}\bra{\psi}(2\tilde{E}_{0|0}-\mathbb{I}_{C})\otimes(2E_{0|0}-\mathbb{I}_{P})\ket{\psi}}\nonumber\\
&\leq\sqrt{\eta}\nonumber
\end{align}
\end{proof}

Note that we have phrased the lemma in terms of the variable $\eta$ and not $\epsilon$ as in the main text of the paper. We can relate the two since if the conditions of $f(\epsilon)$-CST are met then all probabilities differ from the ideal by $\epsilon$, which then implies that, say, $\langle\psi\vert \tau_{z}\otimes Z\vert\psi\rangle=\bra{\psi}(2\tilde{E}_{0|0}-\mathbb{I})\otimes(2E_{0|0}-\mathbb{I})\ket{\psi}\geq 1-8\epsilon$ since each probability incurs an error of $\epsilon$. Putting this value of $\eta=8\epsilon$, we see that our analysis in the above lemma incurs a less favourable constant than in Lem. \ref{niceobs}. However, given the above lemma we may use exactly the same strategy in Thm. \ref{thm1} to obtain a possibility result on self-testing based on the above EPR-steering inequality now in terms of $\eta$.

\begin{proposition}
For the EPR experiment, $f(\eta)$-robust correlation-based one-sided self-testing based on the EPR-steering inequality satisfying $\langle\psi\vert \tau_{z}\otimes Z\vert\psi\rangle+\langle\psi\vert \tau_{x}\otimes X\vert\psi\rangle\geq 2-\eta$ where $f(\eta)=13\sqrt{\eta}$
\end{proposition}

\begin{proof}
The proof essentially follows that of Thm. \ref{thm1} except now we use Lem. \ref{niceobsnew} every time Lem. \ref{niceobs} is used. One difference is now that for $X=2E_{0|1}-\mathbb{I}_{P}$ and for the Pauli-$X$ matrix $\tau_x=2\ket{+}\bra{+}-\mathbb{I}$ we have
\begin{align}
\Vert\mathbb{I}_{C}\otimes X\ket{\psi}-\tau_x\otimes\mathbb{I}_{P}\ket{\psi}\Vert&\leq 2\Vert\mathbb{I}_{C}\otimes E_{1|0}\ket{\psi}-\tilde{E}_{1|0}\otimes\mathbb{I}_{P}\ket{\psi}\Vert -\Vert\ket{\psi}-\ket{\psi}\Vert\nonumber\\
&\leq 2\sqrt{\eta},
\end{align}
and likewise for $Z$ and $\tau_{z}$, the Pauli-$Z$ matrix.

The other difference is in the final stage where we chose $\ket{\mathcal{A}}$ to be the pure state that is proportional to $\ket{0_{C}}\langle{0_{C}}\ket{\psi}$, i.e. $\ket{\mathcal{A}}=\beta^{-\frac{1}{2}}\langle{0_{C}}\ket{\psi}$ where $\beta=\bra{\psi}0_{C}\rangle\langle 0_{C}\ket{\psi}$. We must bound the error associated with making this choice. We use the following observation that
\begin{equation}
\Vert\left(\tau_{z}\tau_{x}-\tau_{z}\otimes X\right)\ket{\psi}\Vert\leq 2\sqrt{\eta}
\end{equation}
which in turn implies that 
\begin{equation}
\Vert\left(-\tau_{x}\tau_{z}-\tau_{z}\otimes X\right)\ket{\psi}\Vert\leq 2\sqrt{\eta}.
\end{equation}
Observing that $\vert\langle{u}\ket{\psi}\vert\leq\epsilon$ if $\Vert\ket{\psi}\Vert\leq\epsilon$ so if we choose $\ket{u}=X\ket{\psi}$ we have that
\begin{align}
\vert\bra{\psi}\tau_{z}\ket{\psi}-\bra{\psi}\tau_{z}\tau_{x}\otimes X\ket{\psi}\vert&\leq 2\sqrt{\eta}\nonumber\\
\vert\bra{\psi}\tau_{z}\ket{\psi}+\bra{\psi}\tau_{z}\tau_{x}\otimes X\ket{\psi}\vert&=\vert\bra{\psi}\tau_{z}\ket{\psi}+\bra{\psi}\tau_{x}\tau_{z}\otimes X\ket{\psi}\vert\leq 2\sqrt{\eta}\nonumber,
\end{align}
where the equality in the second line results from invariance of the absolute value under complex conjugation. Therefore we have
\begin{equation}
2\vert\bra{\psi}\tau_{z}\ket{\psi}\vert\leq\vert\bra{\psi}\tau_{z}\ket{\psi}-\bra{\psi}\tau_{z}\tau_{x}\otimes X\ket{\psi}\vert+\vert\bra{\psi}\tau_{z}\ket{\psi}+\bra{\psi}\tau_{x}\tau_{z}\otimes X\ket{\psi}\vert\leq 4\sqrt{\eta}
\end{equation}
which then implies that $2\vert\bra{\psi}0_{C}\rangle\langle 0_{C}\ket{\psi}-\frac{1}{2}\vert\leq 2\sqrt{\eta}$ and thus $\Vert 2\langle{0_{C}}\ket{\psi}\bra{\psi}0_{C}\rangle-\ket{\mathcal{A}}\bra{\mathcal{A}}\Vert_{1}\leq 2\sqrt{\eta}$. This then completes our proof.
\end{proof}

\section{Demonstrating the optimal trace distance between reference and physical states}\label{app1e}

For the EPR experiment, let us consider the trace distance $D(\ket{\Phi}\bra{\Phi},\ket{\mathcal{A}}\bra{\mathcal{A}}\otimes\ket{\tilde{\psi}}\bra{\tilde{\psi}})$ for all possible isometries $\Phi$ and not just the SWAP isometry. An isometry will take the physical state $\ket{\psi}$ to $U\ket{\psi}\ket{\hat{0}}$ by introducing ancillae $\ket{\hat{0}}$ and applying a unitary $U$ to the physical state and ancillae. As discussed in Sec. \ref{sec1}, the trace distance is then $D(U(\ket{\psi}\bra{\psi}\otimes\ket{\hat{0}}\bra{\hat{0}})U^{\dagger},\ket{\mathcal{A}}\bra{\mathcal{A}}\otimes\ket{\tilde{\psi}}\bra{\tilde{\psi}})=\sqrt{1-F^{2}}$ for $F=\vert\bra{\mathcal{A}}\bra{\tilde{\psi}}U\ket{\psi}\ket{\hat{0}}\vert$. We write $\ket{\psi}$ in terms of its Schmidt decomposition
\begin{equation*}
\ket{\psi}=\sqrt{\lambda}\ket{u}\ket{v}+\sqrt{1-\lambda}\ket{u^{\perp}}\ket{v^{\perp}}
\end{equation*}
for $\lambda$ as some real number such that $0\leq\lambda\leq 1$ and $\langle u^{\perp}\ket{u}=\langle v^{\perp}\ket{v}=0$. Since $\ket{u}$ is a state of a qubit it may be written as $\ket{u}=\cos{\frac{\theta_{1}}{2}}\ket{0}+e^{i\theta_{2}}\sin{\frac{\theta_{1}}{2}}\ket{1}$. Given this, we obtain
\begin{equation*}
F=\frac{1}{\sqrt{2}}\vert\bra{\mathcal{A}}\bra{0}\left(\sqrt{\lambda}\cos{\frac{\theta_{1}}{2}}\ket{w}+\sqrt{1-\lambda}e^{-i\theta_{2}}\sin{\frac{\theta_{1}}{2}}\ket{w^{\perp}}\right)+\bra{\mathcal{A}}\bra{1}\left(\sqrt{\lambda}e^{i\theta_{2}}\sin{\frac{\theta_{1}}{2}}\ket{w}-\sqrt{1-\lambda}\cos{\frac{\theta_{1}}{2}}\ket{w^{\perp}}\right)\vert,
\end{equation*}
where $\ket{w}=U\ket{v}\ket{\hat{0}}$ and $\ket{w^{\perp}}=U\ket{v^{\perp}}\ket{\hat{0}}$. We now maximize $F$ for all isometries so as to obtain a lower bound on $D(\ket{\Phi}\bra{\Phi},\ket{\mathcal{A}}\bra{\mathcal{A}}\otimes\ket{\tilde{\psi}}\bra{\tilde{\psi}})$. The value of $F$ will be maximized when $\ket{w}$ and $\ket{w^{\perp}}$ is in the linear span of $\{\ket{\mathcal{A}}\ket{0},\ket{\mathcal{A}}\ket{1}\}$. Therefore, $\ket{w}=\cos{\frac{\theta_{3}}{2}}\ket{\mathcal{A}}\ket{0}+e^{i\theta_{4}}\sin{\frac{\theta_{3}}{2}}\ket{\mathcal{A}}\ket{1}$ and $F^{*}$ will be the maximum of 
\begin{equation*}
\frac{1}{\sqrt{2}}\vert\left(\sqrt{\lambda}\cos{\frac{\theta_{1}}{2}}\cos{\frac{\theta_{3}}{2}}+\sqrt{1-\lambda}e^{-i(\theta_{2}+\theta_{4})}\sin{\frac{\theta_{1}}{2}}\sin{\frac{\theta_{3}}{2}}\right)+\left(\sqrt{\lambda}e^{i(\theta_{2}+\theta_{4})}\sin{\frac{\theta_{1}}{2}}\sin{\frac{\theta_{3}}{2}}+\sqrt{1-\lambda}\cos{\frac{\theta_{1}}{2}}\cos{\frac{\theta_{3}}{2}}\right)\vert
\end{equation*}
which then implies that $F^{*}=(1/\sqrt{2})(\sqrt{\lambda}+\sqrt{1-\lambda})$. We now wish to put bounds on $\lambda$ which can be easily attained since $\rho_{C}=\lambda\ket{u}\bra{u}+(1-\lambda)\ket{u^{\perp}}\bra{u^{\perp}}$ and $\tilde{\rho}_{C}=\frac{1}{2}\mathbb{I}_{C}=\frac{1}{2}(\ket{u}\bra{u}+\ket{u^{\perp}}\bra{u^{\perp}})$. If we assume that $D(\rho_{C},\tilde{\rho}_{C})=\epsilon$ then we have that $\vert\lambda-\frac{1}{2}\vert=\epsilon$ and thus
\begin{equation*}
F^{*}=\frac{1}{\sqrt{2}}(\sqrt{\frac{1}{2}+\epsilon}+\sqrt{\frac{1}{2}-\epsilon})=1-\frac{1}{2}\epsilon^{2}-O(\epsilon^{3}),
\end{equation*}
where in the last equation we take the Taylor series expansion of $F^{*}$ and $O(\epsilon^{3})$ represents polynomials of degree $3$ and higher. In conclusion, given $\epsilon$-closeness of the reduced states, there is an isometry $\Phi$ such that $D(\ket{\Phi}\bra{\Phi},\ket{\mathcal{A}}\bra{\mathcal{A}}\otimes\ket{\tilde{\psi}}\bra{\tilde{\psi}})\leq O(\epsilon)$. This then demonstrates that our SWAP isometry is not optimal for demonstrating such closeness between physical and reference states. However, the optimal isometry will be dependent on the basis $\{\ket{u},\ket{u^{\perp}}\}$ and thus more complicated than the SWAP isometry.

\end{appendix}
\end{document}